\documentclass[final,5p,times,twocolumn]{elsarticle}
\pdfoutput=1
\usepackage{amsmath,amssymb,amsfonts}
\usepackage{textcomp}

\usepackage{subfigure}
\usepackage{diagbox}
\usepackage{enumitem}
\usepackage{color}
\usepackage{mathrsfs}
\usepackage[linesnumbered,ruled,vlined]{algorithm2e}
\usepackage{amsfonts,amssymb}

\usepackage{amsthm}
\theoremstyle{remark}

\newenvironment{proof}{{\noindent \it Proof.\thinspace}}{\hfill $\blacksquare$\par}
\newtheorem{theorem}{Theorem}
\newtheorem{lemma}{Lemma}
\newtheorem{definition}{Definition}
\newtheorem{proposition}{Proposition}
\newtheorem{remark}{Remark}
\newtheorem{assumption}{Assumption}
\newtheorem{corollary}{Corollary}
\newtheorem{pot}{Proof of Theorem}
\usepackage{array}
\usepackage{stfloats}

\usepackage{setspace}
\usepackage{nomencl}

\begin{document}
\begin{frontmatter}
\title{Distributed Resilient Secondary Control for Microgrids with Attention-based Weights against High-density Misbehaving Agents}

%% Authors
\author[university1]{Yutong~Li}
\ead{liyutong@zju.edu.cn}
\author[university2]{Lili~Wang\corref{au2}}
\ead{wangll@sustech.edu.cn}

\cortext[au2]{Corresponding author.}
%\affiliation[university1]{
%organization={College of Electrical Engineering, Zhejiang University},
%state={Zhejiang},
%country={China}
%}
%\affiliation[university2]{
%organization={School of System Design and Inteligent Manufacturing, Southern University of Science and Technology},
%city={Shenzhen},
%country={China}
%}

\begin{abstract}
Microgrids (MGs) have been equipped with large-scale distributed energy sources (DESs), and become more vulnerable due to the low inertia characteristic. In particular, high-density misbehaving DESs caused by cascading faults bring a great challenge to frequency synchronization and active power sharing among DESs. To tackle the problem, we propose a fully distributed resilient consensus protocol, which utilizes confidence weights to evaluate the level of trust among agents with a first-order filter and a softmax-type function. 
We pioneer the analysis of this nonlinear control system from the system operating range and the graph structure perspectives. Both necessary and sufficient conditions are provided to ensure DACC to be uniformly ultimately bounded, even in a robust network with low connectivity. Simulations on a modified IEEE33-bus microgrid testbed with 17 DESs validate that DACC outperforms existing methods in the presence of 8 misbehaving DESs.
\end{abstract}
%Unlike existing research, this study considers a more challenging condition where the incoming neighbor set of any follower may contain more misbehaving agents than normal agents. 

%\maketitle
\end{frontmatter}

\section{introduction}
A microgrid (MG) is a small-scale power system capable of operating independently or in conjunction with the main power grid. With its local control and power quality support, MGs enable scalable integration of distributed energy sources (DESs)\citep{DCSSPN}. MGs can operate in two modes: grid-connected and islanded, where the islanded mode is more vulnerable to disturbances, such as renewable energy generation or load changes \citep{RCSCIAC}. To ensure stability during islanded operation, hierarchical control, comprising primary, secondary, and tertiary controls, is widely applied in MGs \citep{FDHC}. Primary control maintains voltage and frequency stability, while secondary control restores these parameters to nominal values if deviations occur \citep{DCSC}. Tertiary control facilitates economically optimal operation over longer timescales and calculates reference values for DESs. To enhance reliability, scalability, and communication networks' sparseness in MGs, researchers have extensively studied distributed secondary control architectures \citep{SFVC}, where each DES is considered as an agent. The secondary control is analogous to a leader-follower tracking problem, where all DESs act as followers synchronizing to leaders, represented by the reference values. Leaders are connected to only a few followers \citep{REZAEE2021109384}.

The distributed manner relies on a reliable cyber layer to facilitate information sharing among agents \citep{DRCESS}. Due to the low inertia of DESs, they can respond swiftly to disturbances \citep{SDFLSC}. Consequently, security against misbehaving agents is becoming a critical issue. Misbehaving agents, including faulty, Byzantine, and malicious agents, always exhibit misbehaving states and use unknown update mechanisms \citep{RACR}. Incidents such as the WannaCry ransomware attack \citep{Wannacry} and the Brazil power blackouts \citep{Brazil} have demonstrated the potential for cascading failures or attacks on power systems. A cascading failure occurs when one or more DES failures trigger a sequence of events that ultimately lead to a widespread system collapse \citep{Cascading}. In such cascading failures, high-density misbehaving agents (HDMA) can form within specific segments of the MG, leading to situations where each normal agent may have more misbehaving incoming neighbors than normal ones \citep{SCMS}. Therefore, it is crucial to achieve distributed resilient secondary control against HDMA.

% Resiliency in these problems is to develop distributed control protocols that enable normal agents to achieve consensus amongst themselves and to establish conditions under which such resilient consensus can be attained.% The main difference between them is that at each event time Byzantine agents can send different values to different out-neighbors in a point-to-point network, while malicious agents send the same value to all out-neighbors, regardless of the communication mode. If this difference is not important to the research, the term misbehaving agents can be used to refer to both of them.

To cope with misbehaving events, researchers have proposed three main classes of techniques: 1) The first class consists of algorithms that enable each normal agent to ignore the most deviated agents in the updates, represented by mean-subsequence-reduce (MSR)-type algorithms \citep{RACR}, such as DP-MSR \citep{DPMSR}, SW-MSR \citep{RLFC}, QW-MSR \citep{RRQC} and event-based MSR \citep{RCTEBC}. These algorithms disregard the information received from suspicious agents or those with unsafe values whether or not they are truly misbehaving. In particular, \citep{RVCUC} and \citep{RMDCAE} find the convergence reference by computing an auxiliary point that lies in the convex hull of normal agents' initial states. However, most techniques in this class have certain limitations, such as requiring sorting algorithm execution at every step and high graph connectivity, which can not be satisfied in the network with HDMA. 2) Another common technique is constructing observers to estimate the actual states of neighbor DESs \citep{RSVFAC,OBRIDC,ARDCIST}. These strategies can work in low-connectivity communication network, but often requires false state to be within an upper bound range. An appointed-time observer-based approach in \citep{RCOBA} avoids the network connectivity restriction, but it cannot handle Byzantine attacks. Some strategies \citep{RCSCIAC,FDIARD} require an additional virtual communication network, incurring extra network construction costs. 3) The third class involves assessing the confidence of each neighbor using historical data. RoboTrust in \citep{TMMAC} discards distrusted agents if they deviate from observers. A Q-learning-based consensus algorithm \citep{RLMA} evaluates neighbors based on historical confidences. A stochastic detection compensation based consensus algorithm is developed in \cite{TRAC}. Two types of trust evaluation metrics with different attack indices and time scales are designed in \citep{DEBRSC}. However, most research in this class lacks a straightforward discussion on the robustness and connectivity of communication graphs.

Additionally, many studies have made improvements to tackle resilient leader-follower tracking problems. Misbehaving agents can be regarded as "leaders" attempting to divert normal followers' states to harmful values \citep{RLFC}. To ensure the presence of healthy leaders, the idea of using multiple leaders has been proposed in \citep{multileader}. Unbounded attacks in multi-leader networks are considered in \citep{ROFCHMSA}, but the derivative of the attack is bounded. These techniques still require a communication graph with high connectivity, limiting their effectiveness against HDMA.
 
Motivated by the above discussions, this article addresses a distributed resilient secondary control problem for MGs, formulated as a resilient leader-follower tracking problem. Compared to existing research, the misbehaving events considered here are more challenging, with HDMA potentially present in communication networks with lower connectivity and robustness. These events can impact various components of both followers and leaders, including controllers, actuators, and communication channels. Furthermore, the misbehaving states may have no upper bounds and contain significant channel noise. We propose a discounted attention-type confidence-based consensus (DACC) protocol that employs an exponential confidence function as the weight of the digraph, evaluating the credibility of neighbors' states based on accumulated historical confidence. The main contributions of this article are as follows:
%The convergence process of pur protocol is close to the flocking in Cucker-Smale model\cite{CSFRL}\citep{FCSGD}. 

\begin{enumerate}
\item We develop a discrete-time non-linear distributed protocol to address the resilient secondary control problem against HDMA, focusing on frequency synchronization and active power sharing. The protocol eliminates the need for any additional information beyond the states of neighbors.
%in the case where the lower connectivity condition of the communication graph is satisfied than regular MSR-type algorithms. 

\item We propose a novel one-time design method for the control parameters of DACC based on the expected resilient performance, the connectivity of the graph, and the operating ranges of the system. 
%Our method is resilient to unknown and no upper bounded misbehavior with high-frequency noises generated by both misbehaving leaders and misbehaving followers. 

\item We discuss both the necessary and sufficient conditions for achieving Leader-follower Resilient Uniformly Ultimately Bounded control using the designed parameters. To the best of our knowledge, this article is the first to prove the resilience and convergence of systems with normalized exponential weights from the perspective of graph structure. The effectiveness of DACC is validated by applying it to a modified IEEE 33-bus MG with 17 DESs.
\end{enumerate}

This article is structured as follows. Section II introduces the preliminaries and formulates the research problem. Our consensus protocol and parameter design method are proposed in Section III. Section IV discusses the stability and properties of the proposed algorithm. Section V presents simulation examples to illustrate the design. Section VI concludes this article.

\section{Preliminaries and Problem Formulation}
In this section, some preliminaries on the graph theory and resilience concepts are given, and then the resilient secondary control problem is formulated.
\vspace{-8pt}
\subsection{Notation}
Throughout the paper $\mathbb{R}$, $\mathbb{R}_+$, and $\mathbb{R}_{\geq 0}$ denote the sets of real, positive real, and nonnegative real numbers, respectively; The notations $\mathbb{N}$ and $\mathbb{N}_0$ denote the sets of natural numbers including and excluding zero, respectively. For a scalar $x$, $\vert x \vert$ denotes the absolute value, and for a set $\mathcal{S}$, $ \vert \mathcal{S} \vert$ stands for the cardinality. For two sets $\mathcal{S}_1$ and $\mathcal{S}_2$, the reduction of $\mathcal{S}_1$ by $\mathcal{S}_2$ is denoted by $\mathcal{S}_1 \backslash \mathcal{S}_2$; and for three sets $\mathcal{S}_1$, $\mathcal{S}_2$, and $\mathcal{S}_3$, the reduction of $\mathcal{S}_1\backslash \mathcal{S}_2$ by $\mathcal{S}_3$ is denoted by $\mathcal{S}_1 \backslash \mathcal{S}_2 \backslash \mathcal{S}_3$. For a matrix $ M =[m_{ij}]_{m \times n}$, $\Vert M \Vert$ denotes the infinite norm of $M$, satisfying $\Vert M \Vert=\max_{1\leq i \leq m}\sum_{j=1}^n \vert m_{ij}\vert $. Formally, the standard softmax function $\operatorname{softmax}: \mathbb{R}^K \rightarrow(0,1)^K$, where $K \geq 1$, takes a vector $\mathbf{z}=\left(z_1, \ldots, z_K\right) \in \mathbb{R}^K$ and computes each component of vector $\operatorname{softmax}(\mathbf{z}) \in(0,1)^K$ with
$$
\operatorname{softmax}_i(\mathbf{z})=\frac{\exp(z_i)}{\sum_{j=1}^K \exp(z_j)}
$$
\subsection{Graph theory}
To model an interaction graph among $N$ DESs in a MG, we adopt a digraph (directed graph) $\mathcal{G}=\{\mathcal{V},\ \mathcal{E}\}$ with a finite set of agents $\mathcal{V}=\{1,\ 2, \ \cdots, \ N\}$, and a set of edges $\mathcal{E}\subseteq \mathcal{V} \times \mathcal{V}$. If agent $i$ can receive information from agents $j$, there is a directed edge from $j$ to $i$, that is $(j, i) \in \mathcal{E}$. The neighbor set of agent $i$ is defined by $\mathcal{V}_i^{in}=\{j \in \mathcal{V} \vert (j, i) \in \mathcal{E}\}$. Here the cardinality of $\mathcal{V}_i^{in}$ which is the number of agent $i$'s neighbors is called {\it in-degree} $d_i$. The maximum in-degree of all the agents in $\mathcal{V}$ is denoted as $d_{max}$. An agent $i$ is said to be reachable from an agent $j$ if there exists a directed path from $j$ to $i$. The directed distance from $j$ to $i$, denoted by $\text{dist}(j,i)$, is defined to be the number of edges of the shortest directed paths from $j$ to $i$. Moreover, a nonempty set $\mathcal{S} \subseteq \mathcal{V}$ is {\textit {$r$-reachable}} if $\exists i \in \mathcal{S}$ s.t. $\vert \mathcal{V}_i^{in} \backslash \mathcal{S} \vert \geq r$. 

In this paper, we consider a leader-follower graph where the leaders propagate a desired reference signal for the followers. The set of leaders is defined as $\mathcal{L}$ and the set of followers is defined as $\mathcal{F}$ where $\mathcal{F} \cup \mathcal{L}=\mathcal{V}$, and $\mathcal{F}\cap \mathcal{L}=\emptyset$. 
To model a resilient leader-follower problem, 
we introduce the concepts of graph properties: the depth \citep{CSFRL}, distance-based partition, and $r$-robustness \citep{REZAEE2021109384} for the leader-follower graph $\mathcal{G}$.

\begin{definition}\label{defi_depth} (the depth for a leader-follower graph)
For a digraph $\mathcal{G}$ existing at least one leader that has a directed path to each follower, the depth of $\mathcal{G}$ is the largest distance from a leader in $\mathcal L$ to a follower in $\mathcal{F}$, that is,
\begin{equation}\nonumber
h=\max \{\text{dist} (i, j), i\in \mathcal{L} , j\in \mathcal{F} \}
\end{equation}
\end{definition}

\noindent The distance-based partition is discussed next. The followers can be partitioned according to the distance to the leaders: 
 \begin{equation}\nonumber
\mathcal{V}^{(n)}=\{ j \in \mathcal{F} | \ \min_{i \in \mathcal{L}} \{ \text{dist} (i, j)\} =n\}, \ n=1,2,\cdots, h.
\end{equation} 
\noindent Note that $\mathcal{V}^{(n)}\cap \mathcal{V}^{(n')}=\emptyset$ if $n\neq n'$ for $n,n'\in \{1,2,\ldots,h\}$, and $\cup_{n=1,2,\ldots,h} \mathcal{V}^{(n)}=\mathcal F$.

\begin{definition}($r$-robustness for a leader-follower graph)\label{defi_robust} For a directed graph $\mathcal{G}=\{\mathcal{V},\ \mathcal{E}\}$ with leaders in $\mathcal{L}$, $\mathcal G$ is a $r$-robust leader–follower graph, if $\vert \mathcal{V}^{(1)} \vert \geq r$ and any nonempty set $\mathcal{S} \subseteq \mathcal{V} \backslash \mathcal{L} \backslash \mathcal{V}^{(1)}$ is r-reachable.
\end{definition}

% Depth is related to the leaders' position on the graph, it evaluates the digraph characteristics along with the connectivity from different perspectives. 
\vspace{-8pt}
\subsection{Modelling of MG system} 
In an islanded AC MG consists of several heterogeneous DESs, the set of which are denoted as $\mathcal F$. The control systems of DESs are realized through a hierarchical control structure, consisting of primary, secondary, and tertiary controls. Primary control operates on a fast timescale and maintains the basic frequency stability of the DES. In this level, the frequency is determined by the droop control:
\begin{equation} \label{eq:droop}
\omega_i =\theta_i-m_iP_i,
\end{equation}
where $\omega_i$ and $\theta_i$ are the frequency and the primary control reference of DES $i$. $m_i$ is the droop coefficient and $P_i$ is the active power. Secondary control is deployed to compensate for frequency deviations and achieve active power sharing during fault conditions \citep{DCSC}. The objectives of the secondary control are described as follows:
\begin{itemize}[leftmargin=1em, itemindent=0em]
 \item Synchronization of frequencies to the reference frequency $\omega^l$ received from the tertiary control, 
 \begin{equation} \label{eq:sync_obj}
 \lim\limits_{k \to \infty} (\omega_i(k)-\omega^l)\rightarrow 0,\quad i \in \mathcal{F},
 \end{equation}
 where $k \in \mathbb{N}$ is the secondary control time step.
 \item Proportional active power sharing for maintaining the active power limit of individual DESs, 
 \begin{equation}\label{eq:power_obj}
 \lim\limits_{k \to \infty} (m_i P_i(k)-m_j P_j(k))\rightarrow 0,\quad i,j \in \mathcal{F},
 \end{equation}
 which is equivalent to
 \begin{equation}
 \lim\limits_{k \to \infty} (m_i P_i(k)-\frac{\sum_{j\in\mathcal{F}} P_j(k)}{\sum_{j\in\mathcal{F}} m_j^{-1}})\rightarrow 0,\quad i \in \mathcal{F},
 \end{equation}
\end{itemize}

The total active power can be described by $\sum_{i\in\mathcal{F}} P_i(k)= P_{load}+P_{loss}$, where $P_{load}$ and $P_{loss}$ are the total load demand and power loss of the MG, respectively. Based on data from energy information administration \citep{DSCFAPS}, the power loss can be estimated by approximately $6\%$ of $P_{load}$. Thus, the reference active power can be estimated as $P^l=1.06 P_{load} /\sum_{i\in\mathcal{F}} m_i^{-1}$, which is also provided by the tertiary controller. 

These objectives can be achieved by setting $\theta_i$ for primary control. Utilizing feedback linearization, we have 
\begin{equation}\label{dynamic}
\begin{split}
&\theta_i(k+1)-\theta_i(k)=u_{i}(k) \quad k \in \mathbb{N},\ i \in \mathcal{F}
\end{split} 
\end{equation} 
where $u_i$ is the secondary control input to be designed. 
%Due to geographical and communication distance limitations, the tertiary controller cannot send its reference values to all DESs. 
With a communication network represented by a leader-follower digraph $\mathcal{G}=\{ \mathcal{V},\mathcal{E} \}$, the secondary control input in \eqref{dynamic} can be achieved in a distributed manner. To ensure the security of reference values, the tertiary controller sends the same reference messages $\omega^l$ and $P^l$ to the DESs in $\mathcal{V}^{(1)}$ via multiple channels. Each channel is regarded as a leader in $\mathcal{L}$. Besides, to eliminate the estimation errors of the total active power in the tertiary controller, all DESs in $\mathcal{V}^{(1)}$ are interconnected to form a fully connected sub-graph, which means $(j,i) \in \mathcal{E}$ for all $j,i \in \mathcal{V}^{(1)}$. Since high reliability multi-channel communication has a high cost, other communication links, which do not involve the tertiary controller, only adopt single-channel communication. Then the control inputs are designed as:
\begin{equation} \label{foll_al}
\begin{split}
u_{i}=&-\sum_{j \in \mathcal{F}}
a_{ij}(k)(\omega_i(k)-\omega_j(k)+m_iP_i(k)-m_jP_j(k))\\
&-\sum_{j \in \mathcal{L}} a_{ij}(k)(\omega_i(k)-\omega_j^l+ m_iP_i(k)-P_j^l),
\end{split}
\end{equation} 
where $\omega_j^l$ and $P_j^l$ are the reference values received from the $j$th channel. Note that for any intact leader $j \in \mathcal{L}$, we have $\omega_j^l=\omega^l$ and $P_j^l=P^l$. The weight $a_{ij}(k) =0$ if $(j,i)\notin \mathcal{E}$, $a_{ij}(k)>0$ if $(j, i)\in \mathcal{E}$, and $\sum_{j \in \mathcal{V}} a_{ij}(k)=1$. 
The design of the weight $a_{ij}$ is contingent on the values of $\theta_i(k)$ and $\theta_j(k)$ in our protocol, which will be provided in Section \ref{sectionIIIA}.
\begin{remark}
\begin{figure}[!t]
\centering
\subfigure[A modified IEEE33-bus electrical network.]
{\includegraphics[width=2.8in]{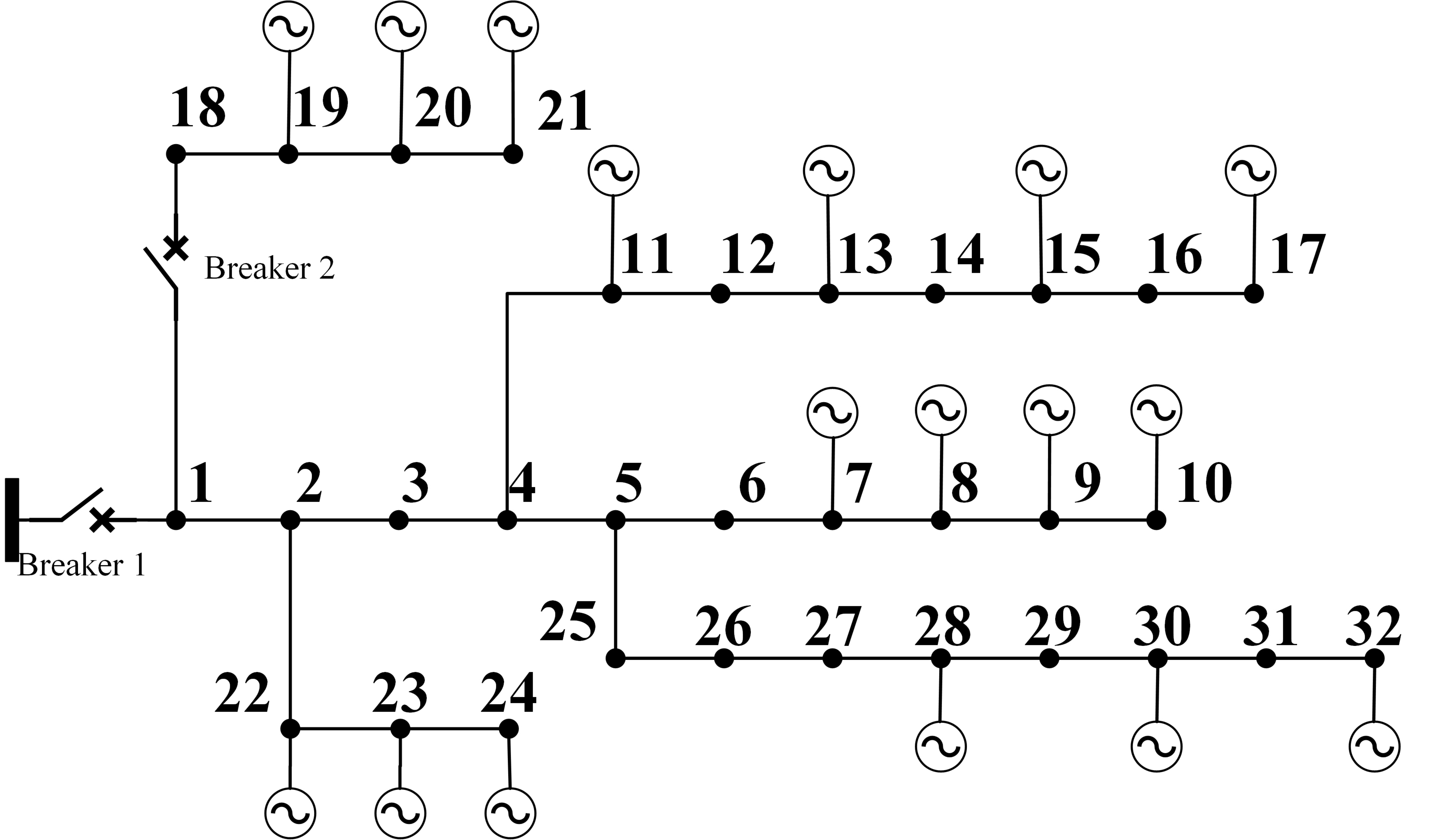}\label{fig:electrical}}
\subfigure[The communication network.]
{\includegraphics[width=2.8in]{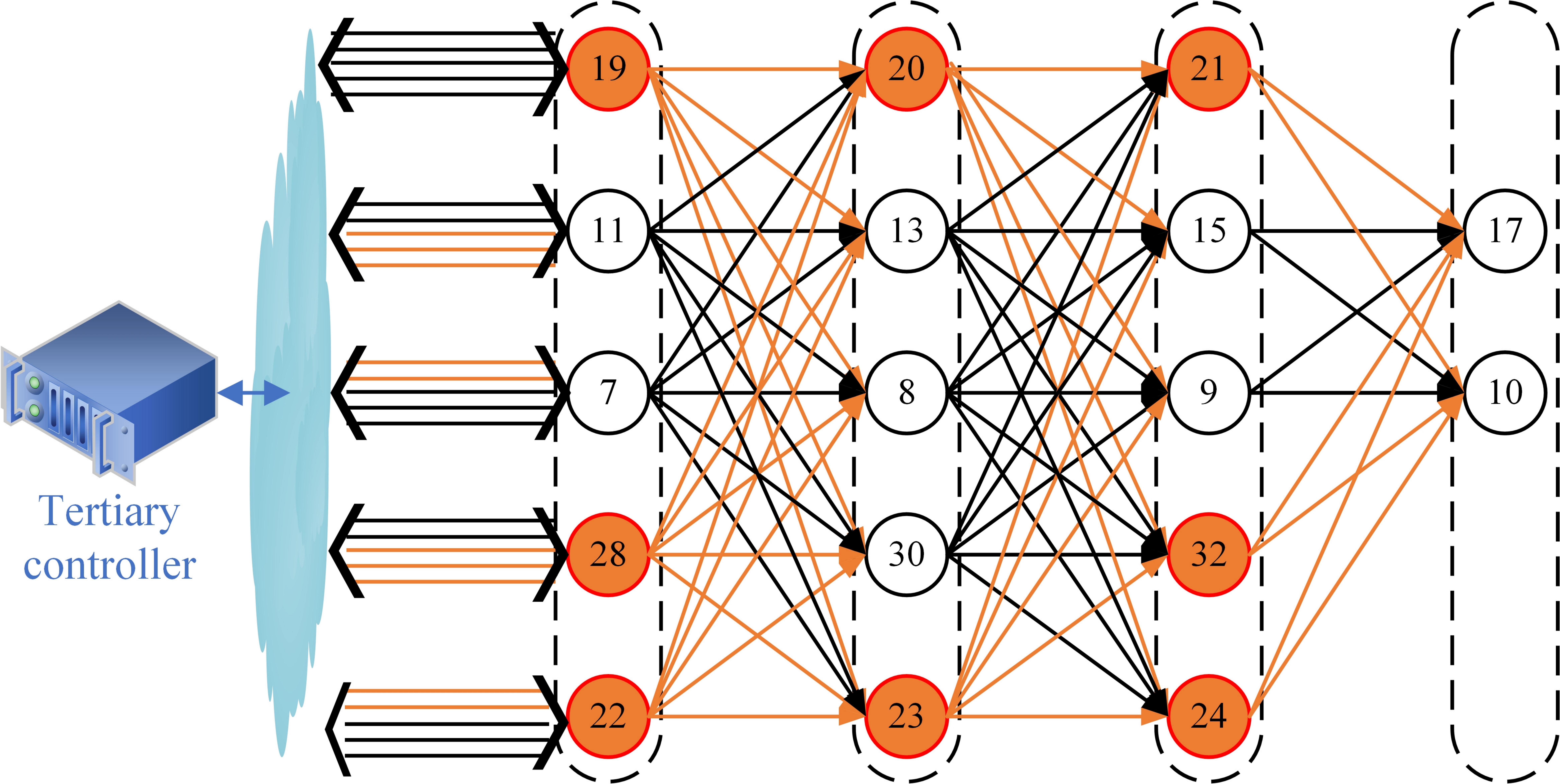}\label{fig:communication}} 
\caption{An example MG system.}\label{fig:MG}
\end{figure}
As shown in Fig.~\ref{fig:MG}, a modified IEEE33-bus electrical network and its communication network are taken as an example. According to Definition~\ref{defi_depth}, the depth of $\mathcal{G}$ is $h=4$ and all DESs can be partitioned into $\mathcal{V}^{(1)}=\{19,11,7,28,22 \}$, $\mathcal{V}^{(2)}=\{30,13,8,30,23 \}$, $\mathcal{V}^{(3)}=\{31,15,9,32,24\}$ and $\mathcal{V}^{(4)}=\{17,10\}$. Each DES in $\mathcal V^{(1)}$ has $5$ neighbors from $\mathcal{L}$ and $4$ neighbors from $\mathcal V^{(1)}$; Each DES in $\mathcal{V}^{(n)}$ has 5 neighbors from $\mathcal{V}^{(n-1)}$ for $n=2,3,4$. According to Definition~\ref{defi_robust}, between any two nonempty disjoint sets, one of them has a agent with at least 5 incoming links from its outside and thus, the graph is a $5$-robust network. The minimum degree for a $r$-robust graph is $r$. 
\end{remark}

\vspace{-6pt}
\subsection{Modeling of misbehaving agents} 
%{\color{cyan}No need to mention the background here, put it in the introduction}
This subsection models the misbehaving agents under reasonable assumptions studied in this study. Two sets of misbehaving agents are considered:
\begin{itemize}[leftmargin=1em, itemindent=0em]
 \item The first set includes DESs with malfunctions in the controller or actuator, represented as $\mathcal{V}_{M1} \subseteq \mathcal{F}$. Their points of common coupling with the MG are cut off by relay protection facilities, and the output power is reduced to 0.
 \item The second set includes misbehaving agents that transmit misbehaving states to others, represented as $\mathcal{V}_{M2}\subseteq \mathcal{V}$. Both followers and leaders may be in this class and for the followers, the internal control loops of the DESs are intact and still undertake power supply tasks of the MG. 
\end{itemize}

Upon the above misbehaving events, the agents can be reclassified into three types of agents: normal leaders, normal followers, and misbehaving agents. The set of $N_L$ normal leaders, $N_F$ normal followers and $N_M$ misbehaving agents are denoted by $ \mathcal{V}_L \subseteq \mathcal{L}, \mathcal{V}_F \subseteq \mathcal{F}$ and $ \mathcal{V}_M =\mathcal{V}_{M1} \cup \mathcal{V}_{M2}$, respectively. Here, $\mathcal{V}_L\cap \mathcal{V}_F\cap \mathcal{V}_M=\emptyset$ and $\mathcal{V}_L\cup \mathcal{V}_F\cup \mathcal{V}_M=\mathcal{V}$. 

\subsubsection{Assumption on topology}Cascading failures tend to occur on DESs, and the misbehaving agents can merge within specific segments of the network. We define that a leader-follower network is with {\it high-density misbehaving agents} if there exists an agent $i \in \mathcal{V}_F$ such that it has more misbehaving neighbors than normal neighbors, i.e., $\vert \mathcal{V}_i^{in} \cap \mathcal{V}_M\vert \geq \vert \mathcal{V}_i^{in} \cap \mathcal{V}_F\vert$. 

\begin{definition}{($f$-local misbehaving)}
A leader-follower network is $f$-local misbehaving if the number of misbehaving agents in the neighborhood of each normal follower is bounded by $f$, i.e., $\forall i \in \mathcal{V}_F, \vert\mathcal{V}_i^{in} \cap \mathcal{V}_M \vert \leq f$.
\end{definition}

%\begin{assumption}\label{as1}
%The leader-follower network is $f$-local misbehaving.
%\end{assumption}
\begin{assumption}\label{as:topology}
In the $f$-local misbehaving leader-follower network, the corresponding graph satisfies that
\begin{itemize}
\item for each agent $i \in \mathcal{V}^{(1)}$, there are at least $g+f$ neighbors from $\mathcal{L}$, i.e., $\vert \mathcal{V}_i^{in}\cap \mathcal{L} \vert \geq g+f$;
\item for each agent $i \in \mathcal{V}^{(n)}\ (n \geq 2)$, there are at least $g+f$ neighbors from $\mathcal{V}^{(n-1)}$, i.e., $\vert \mathcal{V}_i^{in}\cap \mathcal{V}^{(n-1)} \vert \geq g+f$.
\end{itemize}
\end{assumption}
\begin{remark}\label{remark:gf}
Assumption~\ref{as:topology} implies that the leader-follower graph $\mathcal G$ is $(g+f)$-robust. To see this, note that by Definition \ref{defi_robust}, $\mathcal G$ is $(g+f)$-robust if: 1) $\vert \mathcal{V}^{(1)} \vert \geq g+f$, which can be directly observed from Assumption~\ref{as:topology}; 2) Any nonempty set $\mathcal{S} \subseteq \mathcal{F} \backslash \mathcal{V}^{(1)}$ is $(g+f)$-reachable, which is equivalent to the existence of $i\in\mathcal{S}$ such that $|\mathcal{V}_i^{in}\backslash \mathcal{S}|\ge g+f$. To prove the second condition, let $n^* = \min_n \{ n \vert \mathcal{V}^{(n)} \cap \mathcal{S} \neq \emptyset , n\in [2,h]\}$, which implies $\mathcal{V}^{(n^*-1)} \cap \mathcal{S} = \emptyset$.
Therefore, for any $i \in \mathcal{V}^{(n^*)} \cap \mathcal{S}$, it holds that $ \mathcal{V}_i^{in} \cap \mathcal{V}^{(n^*-1)} \subseteq \mathcal{V}_i^{in} \backslash \mathcal{S}$, which, together with $\vert \mathcal{V}_i^{in} \cap \mathcal{V}^{(n^*-1)} \vert \geq g+f$ from Assumption~\ref{as:topology}, implies $\vert \mathcal{V}_i^{in} \backslash \mathcal{S} \vert \geq \vert \mathcal{V}_i^{in} \cap \mathcal{V}^{(n^*-1)}\vert \geq g+f$. Concluding all the above, $\mathcal G$ is $(g+f)$-robust. 
For most existing work\cite{DPMSR,RLFC,RRQC,RCTEBC}, a $(2f+1)$-robust $\mathcal G$ is necessary. Compared to them, Assumption \ref{as:topology} implies the possibility of relaxing the connectivity requirement in the context of HDMA.
\end{remark}
\subsubsection{Assumption on misbehaving states}To resist the above HDMA, we need to make additional assumptions on the misbehaving state range based on real-world scenarios. For example, in an AC three-phase systems with a nominal angular frequency of $100\pi\ rad/s$, according to Chinese national standards, the allowable frequency deviations of power supply are: 1) $\pm 0.4 \%$ for power grid capacity above 3000MW; 2) $\pm 1 \%$ for power grid capacity below 3000MW; 3) $\pm 2 \%$ under irregular operations. 

For the convenience of analyse, define the state ``error" of normal followers as $e_i(k)=\theta_i(k)-\omega^l-P^{lm}$, for $i \in \mathcal{V}_F$, and the state ``error'' of misbehaving agents as $e_i^m(k)=\theta_i(k)-\omega^l-P^{lm}$, for $i \in \mathcal{V}_M$, where $P^{lm}$ is the reference power excluding DESs in $\mathcal{V}_{M1}$, i.e., $P^{lm}=(\sum_{i\in \mathcal{F} \backslash \mathcal{V}_{M1}} P_i)/(\sum_{i\in \mathcal{F} \backslash \mathcal{V}_{M1}} m_i^{-1})$. Based on these, define the safety interval as $\mathcal{I}_S= [-c_n,c_n]$ with a bound $c_n \in \mathbb{R}_+$, which is the maximum value of the initial state`` error" of normal agents, i.e., $c_n=\max_{i\in \mathcal{V}_F}{\vert e_i(0)\vert}$. And for the misbehaving state`` error" beyond the allowable range, we define a lower bound as $c_m \in \mathbb{R}_+$, i.e., $c_m=\min_{i\in \mathcal{V}_M}{\vert e_i^m(k)\vert}$. In this example case, the bound $c_n$ should be smaller than $100\pi \cdot 1 \% \ rad/s$ or $100\pi \cdot 0.4 \% \ rad/s$, and $c_m$ should be larger than $100\pi \cdot 2 \% \ rad/s$. 

However, we should not directly assume all misbehaving state `` errors" outside $(-c_m, c_m)$, because attackers tend to apply false channel noises to interfere with system operation and disguise misbehaving states as normal states. Therefore, we need to conduct the following analysis. The square of the $i$th misbehaving agent's ``error" sequence $(e_i^m(k))^2,\cdots,(e_i^m(k+K_n-1))^2$ is decomposed into a pile of components with different frequencies and amplitudes by Discrete Fourier Transformation,
 \begin{equation}
(e_i^m(\kappa))^2=\chi_i^0 +\sum_{\omega=\omega_0}^{(K_n-1) \omega_0}\chi_i^\omega \exp(\mathrm{j} \omega t_s \kappa), \kappa \in [k,k+K_n)
 \label{DFT}
 \end{equation} 
where $\mathrm{j}$ is the imaginary unit; $K_n$ is the sampling points of window or component number; $t_s$ is the sampling interval of the system~\eqref{dynamic}; $\omega_0=2\pi/(K_n t_s)$ is the fundamental angular frequency; $\chi_{i}^\omega$ is the amplitude corresponding to the angular frequency $\omega$. To deal with noises, we can deploy the following first-order low-pass digital filter,
\begin{equation}\label{eq:filter}
y(k)= \eta (e_{i}^m(k))^2 +(1-\eta) y(k-1), 
\end{equation}
where $\eta \in (0,1)$ is a filter parameter; $(e_{i}^m(k))^2$ is the input and $y$ is the output mitigating the high-frequency noise. The inertia time constant of the filter in the frequency domain is given by $(\eta^{-1}-1)t_s$, and the maximum amplitude with angular frequency $\omega$ of the filter is $[1+((\eta^{-1}-1)t_s \omega)^2]^{-\frac{1}{2}}$. For each Fourier component of the input corresponding to $\omega \in [\omega_0, (K_n-1)\omega_0]$ in \eqref{DFT}, after passing through the filter, the maximum amplitude of this Fourier component signal is $\chi_{i}^\omega[1+((\eta^{-1}-1)t_s \omega)^2]^{-\frac{1}{2}}$. Recombining all Fourier components, the minimum amplitude of the filter's output is $\chi_{i}^0-\sum_{\omega=\omega_0}^{(K_n-1) \omega_0}\chi_{i}^\omega[1+((\eta^{-1}-1)t_s \omega)^2]^{-\frac{1}{2}}$.

%Our study considers misbehaving states beyond the allowable range. We require no upper bound of $e_{i}^m$, but a predetermined lower bound of $y$. The bound is denoted by $c_m \in \mathbb{R}_+$, and should satisfy $c_m>100\pi \cdot2 \% \ rad/s$ in the example. 
Based on the above discussion, we can assume a predetermined lower bound of $y$.
\begin{assumption} \label{as:misbehaving}
For any $c_n \in [0,\infty)$, there exists a bound $c_m>c_n$ and a filter parameter $\eta \in (0,1)$, such that the Fourier components of the misbehaving state ``errors'' satisfy 
$$\chi_{i}^0-\sum_{\omega=\omega_0}^{(K_n-1) \omega_0}\chi_{i}^\omega[1+((\eta^{-1}-1)t_s \omega)^2]^{-\frac{1}{2}} \geq c_m^2.$$
\end{assumption}
\begin{remark}
This assumption means that although misbehaving agents may also produce normal states at some time, but they still need to be isolated. An example of the states of a misbehaving agent can be seen in Fig.~\ref{fig:noise}. In addition, Assumption~\ref{as:topology} does not imply that misbehaving agents can be identified based on the safety interval. Note that normal followers do not know whether the received reference values are intact and consequently, cannot directly determine whether the states of their neighbors fall within the safe interval. 
\end{remark}

\subsection{Problem formulation}
Considering the above misbehaving agents, the system~(\ref{dynamic}) with (\ref{foll_al}) can be expressed as an ``error" system:
\begin{equation}
 \mathbf{e}(k+1)=A(k) \mathbf{e}(k)+ L_m(k)\mathbf{e}^m(k)- L_a(k) \mathbf{e}(k),
\label{eq:error1}
\end{equation}
where $\mathbf{e}(k)=[e_i(k)]_{N_F \times 1}$ and $ \mathbf{e}^m(k)=[e_i^m(k)]_{N_M \times 1}$ are the ``error" vectors of all agents in $\mathcal{V}_F$ and $\mathcal{V}_M$, respectively. $A(k)=[a_{ij}(k)]_{N_F \times N_F}$, while $L_m(k)=[ a_{ij}(k)]_{N_F \times N_M}$, $j \in \mathcal{V}_M$ and $L_a(k)=\text{diag}[ \sum_{j \in \mathcal{V}_{i}^{in} \cap \mathcal{V}_M} a_{ij}(k) ]_{N_F \times N_F}$. Obviously, the objectives of secondary control are transformed into $ \lim\nolimits_{k \to \infty}\Vert\mathbf{e}(k) \Vert\rightarrow 0$ in the system \eqref{eq:error1}.

Our research problem is described based on the concept of bounded consensus for a leader-follower network.

\begin{definition}(LRUUB)\label{RAC}
Consider a leader-follower network with $\mathcal{G}$.
For any agents in $\mathcal V_F$ with initial states in the safety interval $\mathcal{I}_S$, they are said to achieve leader-follower resilient uniformly ultimately bounded (LRUUB) with respect to the reference values, if the following conditions hold:

$\bullet$ Safety condition: $e_i(k) \in \mathcal{I}_S$ for all $i \in \mathcal{V}_F$ and $k \in \mathbb{N}$.

$\bullet$ Cooperative uniformly ultimately bounded (UUB) condition: for any given $\epsilon \in (0,c_n)$, $\exists k_f\in \mathbb{N}$, such that $e_i (k) \in \mathcal{I}_U$ for all $k > k_f, i \in \mathcal{V}_F$, where $\mathcal{I}_U= [-\epsilon,\epsilon]$. 
\end{definition}

\noindent \emph{Problem}: Under Assumption~\ref{as:topology} and \ref{as:misbehaving}, given a directed leader-follower graph $\mathcal{G}=(\mathcal{V},\ \mathcal{E})$, propose a distributed consensus protocol by designing $a_{ij}(k)$ in (\ref{foll_al}), such that the system ~\eqref{eq:error1} can achieve Leader-follower Resilient Uniformly Ultimately Bounded (LRUUB) for the network with HDMA. 

%To cope with misbehaving agents, leaders should maintain their control ability or avoid losing it completely. This implies that the system should have multiple leaders, unless the only leader is completely secure. Moreover,
\section{Resilient Protocol Design}
In this section, a discounted attention-type confidence-base consensus (DACC) protocol is designed to address our problem. It is considered in two aspects, the design of the weights $a_{ij}$, and the design of the parameters involved in designing the weights.
\subsection{Weight design} \label{sectionIIIA}
In the DACC protocol, we apply softmax function to calculate the confidence weights in \eqref{eq:aij}. The softmax function is a classic classification function used in attention mechanism \cite{NMTJL}, which can help each follower focus on normal states and ignore misbehaving states.  For any agent $i \in \mathcal{V}_F$, the weights $a_{ij}(k)$ are designed as 
\begin{equation}\label{eq:aij}
 a_{ij}(k) = \operatorname{softmax}_j (-\sigma \mathbf{r}_i(k)),\ {j \in \mathcal{V}_i^{in}}, 
\end{equation}
where $\sigma\in \mathbb{R}_+$ is a scalar parameter to be designed, and $\mathbf{r}_i(k) \in \mathbb{R}^{\vert \mathcal{V}_i^{in} \vert }$ is the vector of evaluation functions. Each entry of $\mathbf{r}_i(k)$ is denoted by $r_{ij}(k)$, which is an evaluation function of $\theta_{i}(k)-\theta_j(k)$ $j \in \mathcal{V}_i^{in}$, abbreviated as $\theta_{ij}(k)$. The function $r_{ij}(k)$ is defined by
\begin{equation}\nonumber
\begin{aligned}
&r_{ij}(k)=
\left\{
\begin{aligned}
&\theta_{ij}(0)^2 & {k=0},\\
&r_{ij}(k-1)+\eta [\theta_{ij}(k)^2-r_{ij}(k-1)] & {k>0}.\\
\end{aligned} \right.
\end{aligned}
\end{equation}
The evaluation function $r_{ij}$ is the discounted accumulation of the historical state of $\theta_{ij}$, which is a standard first-order low pass filter form as shown in \eqref{eq:filter}.

For the convenience of the foregoing analyse, we defined a confidence function $q_{ij}(k)$:
$$q_{ij}(k)=\exp(-\sigma r_{ij}(k)).$$ And let $q_i^s(k)$ represent the sum of all neighbors' confidence functions: $q_i^s(k)=\sum_{p \in \mathcal{V}_i^{in}} q_{ip}(k)$. Then we can rewrite \eqref{eq:aij} as
\begin{equation}\nonumber
 a_{ij}(k) = \frac{q_{ij}(k)}{q_i^s(k)},\ {j \in \mathcal{V}_i^{in}}.
\end{equation}

The main feature of this algorithm lies in its simplicity. There is no need for information more than that of each agent’s neighbors. 
%Each normal follower reduces the confidence of its potential misbehaving neighbors only referring to the information received from them. 
The parameters $\{\sigma, \ \eta \}$ jointly influence the convergence and the resilience of DACC. As such, we introduce a one-time parameter design method (Algorithm~1) in Section~\ref{sec:parameter}.

\subsection{Parameter design} \label{sec:parameter}
Algorithm~\ref{alg} is provided to design the parameters $\{\sigma,\eta\}$, which can guarantee LRUUB based on Proposition~\ref{prop:sigma} and Theorem~\ref{theorem2} presented in Section \ref{sec:main}. We first define two following functions 
\begin{equation}\nonumber
\begin{aligned}
F_t(x)&=\ln(\frac{hf(c_m-x)}{x})+(h-1)\ln(\overline{q_i^s})-h \ln(g),\\
F_m(x)&=c_m^2-2c_m x-(h-1)x^2, 
\end{aligned}
\end{equation}
where $\overline{q_i^s} =d_{max}-f +f \exp(-\underline{\sigma}_p (c_m-c_n)^2)
$, and $\underline{\sigma}_p =\max \{ \frac{1}{2(c_m-c_n)^2} , \frac{\sqrt{6}+3}{2(c_m-\epsilon)^2} \}$. Then three intervals are defined
\begin{equation} \label{eq:sp}
 \mathcal{I}_\sigma^p=[\underline{\sigma}_p , +\infty),
\end{equation}
\begin{equation} \label{eq:sf}
\mathcal{I}_\sigma^f=\{\sigma \in \mathbb{R}_+ | F_m(c_n) \sigma \geq F_t(c_n)\ \text{and}\ F_m(\epsilon) \sigma \geq F_t(\epsilon)\},
\end{equation}
\begin{equation} \label{eq:ef}
\begin{aligned}
\mathcal{I}_\eta^f=&[\max \{
\frac{(1+\sqrt{h})^2}{2\sigma h c_m^2}, \frac{1}{2\sigma (c_m-c_n)^2},\\ &\frac{\sigma^{-1}F_t(\epsilon)-F_m(c_n)}{(c_n-\epsilon)[2c_m+(h-1)(c_n+\epsilon)]}\},1). 
\end{aligned}
\end{equation}

\begin{algorithm}[htp]
 \caption{Parameter design.}\label{alg}
 \SetAlgoLined
 \SetKwInOut{Input}{Input}
 \SetKwInOut{Output}{Output}
 \DontPrintSemicolon
 \Input{The parameters of $\mathcal{G}$ ($d_{max}, g,f$ and $h$);The parameters of intervals ($\epsilon$, $c_n$ and $c_m$).}
 \Output{The parameters of DACC protocol$\{\sigma,\eta\}$.}
 Calculate $\mathcal{I}_\sigma^p$ in \eqref{eq:sp}, and $\mathcal{I}_\sigma^f$ in \eqref{eq:sf}.
 
 \eIf{$\mathcal{I}_\sigma^f \cap \mathcal{I}_\sigma^p\neq \emptyset$}
 {Select a $\sigma \in \mathcal{I}_\sigma^f \cap \mathcal{I}_\sigma^p$.\\
 Select an $\eta \in \mathcal{I}_\eta^f$ in \eqref{eq:ef}.
 }
 {
 Re-select input parameters $\{g,f,h,c_m,c_n\}$:\\
 Increasing $g,c_m$, or decreasing $f,h,c_n$.
 }
\end{algorithm}

\begin{remark}
If the system~\eqref{eq:error1} meets Assumption~\ref{as:topology} and can achieve LRUUB, then for $\sigma \in \mathcal{I}_\sigma^p$, we have $|e^m-e|\geq c_m-c_n$. Therefore, $\max_{i\in \mathcal{V}_F,j \in \mathcal{V}_M} q_{ij}=\exp(-\sigma_p(c_m-c_n)^2)$. Under Assumption~\ref{as:topology}, the worst damage case for the normal agent occurs when it has $f$ misbehaving agents in its neighborhood, and in that scenario, the maximum of $q_i^s$ is $\overline{q_i^s}$.
\end{remark}

\begin{remark}
% Note that with provided parameters of $\{d_{\max}, g, f, h, \epsilon,c_n, c_m\}$, it is not guaranteed that $ \mathcal{I}_\sigma^p \cap \mathcal{I}_\sigma^f =\emptyset$.
Note that it is not guaranteed that $ \mathcal{I}_\sigma^p \cap \mathcal{I}_\sigma^f \neq \emptyset$ for any provided input parameters. If $\mathcal{I}_\sigma^p \cap \mathcal{I}_\sigma^f =\emptyset$, we have the flexibility to re-select either the parameters of the graph $\mathcal G$ $\{g, f, h\}$ by justifying the communication network, or the safety region related bounds $\{c_m,c_n\}$ by adjusting the system's operational interval.
The extension to the condition on the non-emptiness of $\mathcal{I}_\sigma^p \cap \mathcal{I}_\sigma^f$ is a future direction. While not conceptually difficult, the details may be intricate.
 % {\color{red} Not clear, you may want to rewrite or delete} When dealing with a graph characterized by low robustness or high depth, the leader's control over a normal agent $i$ at the far end $i \in \mathcal{V}^{(h)}$ may be impeded by misbehaving agents in $\mathcal{V}_i^{in}$. In addition, given the proximity of the misbehaving state to the leader's state, the algorithm may face challenges distinguishing between the two states. The interval $\mathcal{I}_\sigma^f$ reflects the above limitations directly. Hence, from a real-world perspective, the existence of $\sigma \in \mathcal{I}_\sigma^f \cap \mathcal{I}_\sigma^p$ is likely to hold in many cases (though not all), warranting a separate study. 
\end{remark}
\vspace{-12pt}
\section{Main Results}\label{sec:main}
In this section, we first provide the results of the network without misbehaving agents. Both necessary and sufficient conditions for LRUUB are discussed. Corollaries to relax graph connectivity assumptions are provided. 
\subsection{Convergence without misbehaving agents}
The convergence of the system without misbehaving agents will be proved in advance. The analysis is not trivial as $a_{ij}(k)$ is a time-varying and nonlinear function of $r_{ij}$.
\begin{theorem}\label{thm:1} Suppose that Assumption~\ref{as:topology} hold with $f=0$. For any $c_n \in [0,\infty)$, $\sigma \in \mathbb{R}_+$ and $\eta \in (0,1]$, the system~(\ref{dynamic}) with \eqref{foll_al} and \eqref{eq:aij} is asymptotically stable at the reference values when there is no misbehaving agents in the network.
\end{theorem}
Without misbehaving agents, the ``error'' system~\eqref{eq:error1} is transformed into
\begin{equation}
 \mathbf{e}(k+1)=A(k) \mathbf{e}(k).
\label{eq:error0}
\end{equation}
By proving that the system~\eqref{eq:error0} is asymptotically stable at the equilibrium point $\mathbf{e} = \mathbf{0}$ ($\mathbf{0}=[0]_{N_F \times 1}$) using Theorem~5.14 in \citep{LTDT}. The proof of Theorem~\ref{thm:1} is in Appendix.
\begin{lemma}\citep{LTDT}\label{lemmades}
Let $ \textbf{e} = \textbf{0}$ be an equilibrium point for the nonlinear autonomous system 
\[\textbf{e}(k + 1) = F(k, \textbf{e}(k))\]
where $F: \mathbb{N}_0 \times \mathcal{I}\rightarrow \mathbb{R}^{N}$ is locally Lipschitz in $\mathbf{e}$, %on $\mathbb{N} \times \mathcal{I}$, 
$\mathcal{I}=\{\textbf{e} \in \mathbb{R}^N| \Vert \textbf{e} \Vert \leq b\}$ for some $b\in\mathbb{R}_{+}$. Suppose that the Jacobian matrix $J(k, \textbf{e})=\frac{ \partial F(k, \textbf{e})}{ \partial \textbf{e}}$ is bounded and Lipschitz on $\mathcal{I}$, uniformly in $k$. 
% Let \[J(k,\textbf{0})=\left[\frac{ \partial F}{ \partial \textbf{e} }(k, \textbf{e})\right]\bigg\vert_{ \textbf{e}=\textbf{0}}.\] 
Then, the origin is an exponentially stable equilibrium point for the nonlinear system if it is an exponentially stable equilibrium point for the linear system $ \textbf{e}(k + 1) = J(k,\textbf{0}) \textbf{e}(k)$.
\end{lemma}
\subsection{Necessary condition for LRUUB}
In this subsection, the resilient mechanism of DACC for system~(\ref{eq:error1}) with misbehaving agents is intuitively illustrated by deriving the necessary conditions for LRUUB.
\begin{figure}[!b]
\centering
\includegraphics[width=2.6in]{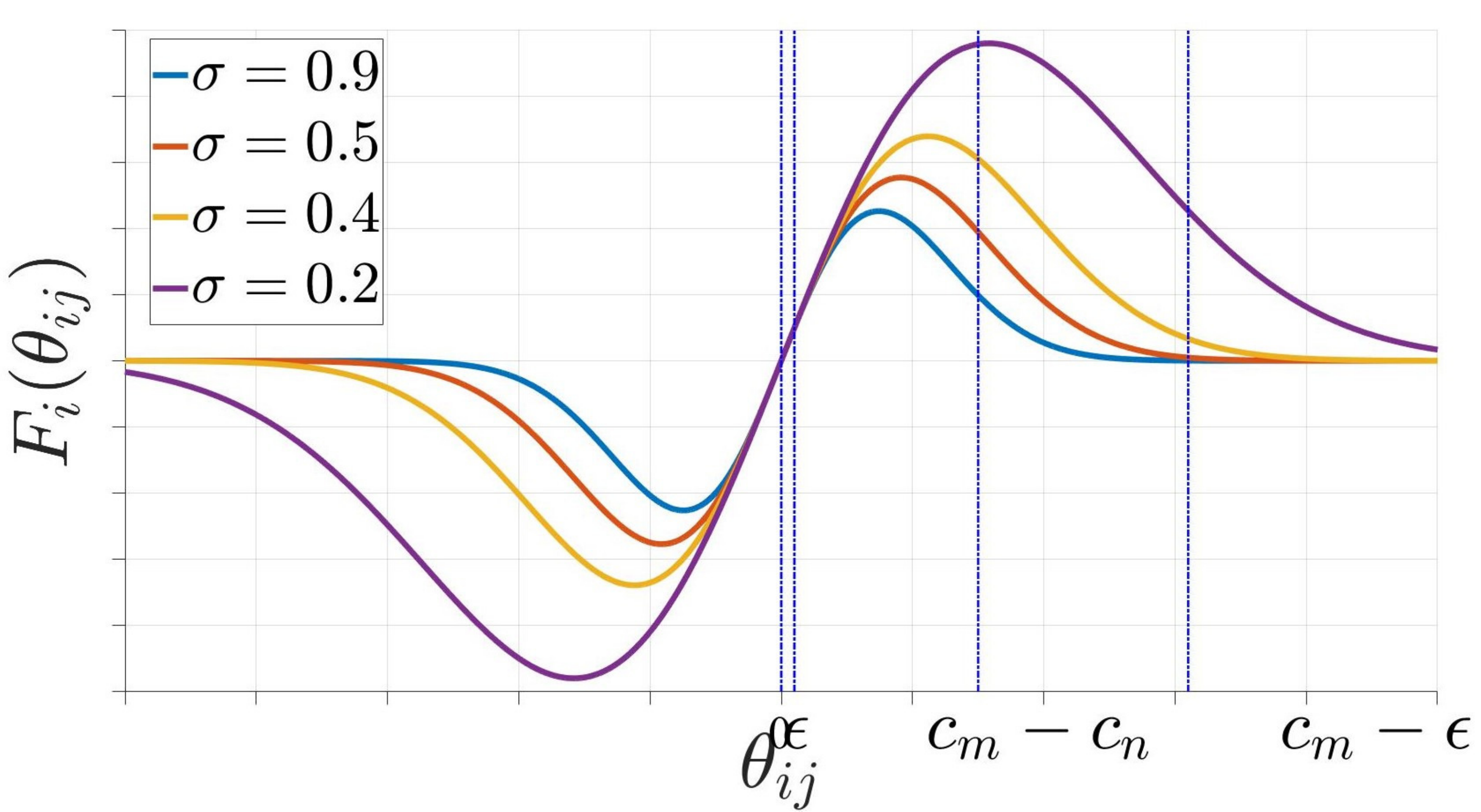}
\caption{Influence function $F_i(\theta_{ij})$.}
\label{fi}
\end{figure}

The resilient mechanism of DACC relies on attenuating the influence of misbehaving states through $a_{ij}\theta_{ij}$ for $i\in \mathcal{V}_F, j \in \mathcal{V}_M$ in the control input $u_i$ defined in \eqref{foll_al} so that the normal agents are not misled by misbehaving agents. We introduce an influence function of $\theta_{ij}$, denoted by 
\begin{equation}\label{eq:fi}
F_i(\theta_{ij})= \exp(-\sigma \theta_{ij}^2) \theta_{ij}.
\end{equation}
Note that $F_i(\theta_{ij})$ is a scaled function of $a_{ij}\theta_{ij}$ when $\eta=1$, i.e., $q_i^sF_i(\theta_{ij})=a_{ij}\theta_{ij}$. The function $F_i$ is a bounded continuous odd function with two poles $((2\sigma)^{-\frac{1}{2}}$, $(2\sigma)^{-\frac{1}{2}} \exp(-\frac{1}{2}))$ and $(-(2\sigma)^{-\frac{1}{2}}, -(2\sigma)^{-\frac{1}{2}} \exp(-\frac{1}{2}))$. As depicted in Fig.~\ref{fi}, if $\vert x_{ij}\vert > (2\sigma)^{-\frac{1}{2}}$, a larger misbehaving neighbor's error can actually decay the influence on the agent.

Further, we say the misbehaving influence term $|a_{ij}\theta_{ij}|$, $i\in \mathcal{V}_F,\ j \in \mathcal{V}_M$ {\it decay rapidly} if $F_i'(|\theta_{ij}|) < 0$, $F_i''(|\theta_{ij}|) > 0$, and $F_i^{'''}(|\theta_{ij}|) < 0$. This implies that a misbehaving neighbor's error can accelerates the decay rate of the influence on the agent.
%This implies that for a larger misbehaving error, $|a_{ij}\theta_{ij}|$, $i\in \mathcal{V}_F,\ j \in \mathcal{V}_M$ in the system ~(\ref{dynamic}) with \eqref{foll_al} and \eqref{eq:aij} becomes smaller.
% Now the necessary conditions for LRUUB is proposed as follows.

\begin{proposition} \label{prop:sigma}
Under Assumption~\ref{as:topology} with $f>0$ and \ref{as:misbehaving}, if the system~(\ref{dynamic}) with \eqref{foll_al} and \eqref{eq:aij} can achieve LRUUB, the parameter $\sigma$ satisfies that:

\begin{itemize}
\item The inequality $\sigma \geq \frac{1}{2(c_m-c_n)^2}$ holds, the misbehaving influence term $|a_{ij}\theta_{ij}|$ for any $i \in \mathcal{V}_F$, $j \in \mathcal{V}_M$ decay;
\item If the inequality $\sigma \geq \frac{\sqrt{6}+3}{2(c_m-\epsilon)^2}$ holds, the misbehaving influence term $|a_{ij}\theta_{ij}|$ for any $i \in \mathcal{V}_F$, $j \in \mathcal{V}_M$ decay rapidly after the system \eqref{dynamic} achieves UUB.
\end{itemize}
\end{proposition}

\begin{proof}
1) We will prove that $\sigma \geq \frac{1}{2(c_m-c_n)^2}$. 
We start with the error dynamics which is shown in \eqref{eq:error1} and it must hold that
\begin{equation}\label{eq:pro1}
\Vert L_m(k)\mathbf{e}^m(k)- L_a(k) \mathbf{e}(k)\Vert = \Vert \mathbf{e}(k+1) - A(k) \mathbf{e}(k)\Vert,
\end{equation}
Under Assumption~\ref{as:topology} and \ref{as:misbehaving}, the system~(\ref{dynamic}) with \eqref{foll_al} and \eqref{eq:aij} can achieve LRUUB, that is there must exist one $\sigma$ such that \eqref{eq:pro1} holds for $\Vert \mathbf{e}(k) \Vert \leq \epsilon$, $k\geq k_f$.

Suppose, by contradiction, $c_m-c_n< (2\sigma)^{-\frac{1}{2}}$. According to Theorem~\ref{thm:1}, we have all the eigenvalues of $A(k)$ within the unit open circle in the complex plane and thus, $\Vert \mathbf{e}(k+1) - A(k) \mathbf{e}(k)\Vert \leq \Vert \mathbf{e}(k+1) \Vert + \Vert A(k) \mathbf{e}(k)\Vert <2\epsilon$ for all $k\geq k_f$. In the $f$-local misbehaving network, we have $\Vert L_m(k)\mathbf{e}^m(k)- L_a(k) \mathbf{e}(k)\Vert \leq f \max_{i,j} |a_{ij}(k)\theta_{ij}(k)|.$
For any $e_i(k+1),\ e_i(k)$ in $\mathcal{I}_U$, $k\geq k_f$, to guarantee (\ref{eq:pro1}) always possesses a solution of parameter $\sigma$, we need to ensure 
\begin{equation}\label{eq:pro2}
f \max_{i,j} |a_{ij}(k)\theta_{ij}(k)|\leq 2\epsilon
\end{equation}
always possesses a solution of $\sigma$. Since the system~(\ref{eq:error1}) achieves LRUUB, we have $\Vert \mathbf{e}(k) \Vert \leq c_n$. Under Assumption~\ref{as:topology}, if $c_m-c_n< (2\sigma)^{-\frac{1}{2}}$, there may exist $|e_j^m(k')| > c_m$ such that $$F_i(|\theta_{ij}(k')|)=(2\sigma)^{-\frac{1}{2}} \exp(-\frac{1}{2}) >F_i(c_m-c_n).$$ Applying it into \eqref{eq:pro2}, we have $$f \max_{i,j} |a_{ij}(k)\theta_{ij}(k)|= f (q_{i,s}(k'))^{-1} F_i(|\theta_{ij}(k')|) \leq 2\epsilon.$$ 
The solution of $\sigma$ locates in 
$$[\frac{1}{8}[(f \exp(-\frac{1}{2}))/(\epsilon q_{i,s}(k'))]^2, [2(c_m-c_n)^2]^{-1}),$$
which is an empty interval for a small $\epsilon$. Therefore, $c_m - c_n \geq (2\sigma)^{-\frac{1}{2}}$, which is, $\sigma \geq \frac{1}{2(c_m-c_n)^2}$ holds. 

2) We will demonstrate that if $\sigma \geq \frac{\sqrt{6}+3}{2(c_m-\epsilon)^2}$ holds, the system~\eqref{dynamic} exhibits excellent steady-state resilient performance. Specifically, the term $|a_{ij}\theta_{ij}|$ will decay rapidly after the system~\eqref{dynamic} achieves UUB. In this case, we directly consider the case where $|\theta_{ij}(k)|=|e_j^m(k) - e_i(k)| \geq c_m - \epsilon$ for all $k \geq k_f$, which is the same as the proof in 1).

Calculate that $F_i'(x)=(1-2\sigma x^2)\exp(-\sigma x^2)$ with three poles $(0,1)$ and $(\pm \sqrt{\frac{3}{2\sigma}}, -2\exp(-\frac{3}{2}))$. As long as $c_m - \epsilon \geq \sqrt{\frac{3}{2\sigma}}$, we have $F_i'(|\theta_{ij}|) \leq F_i'(c_m - \epsilon) < 0$. Similarly, as long as $c_m-\epsilon \geq (\sqrt{6}+3)/(2\sigma^2)$, we have $F_i''(|\theta_{ij}|) \geq F_i''(c_m - \epsilon) > 0$ and $ F_i'''(|\theta_{ij}|) \leq F_i^{'''}(c_m - \epsilon) \leq 0$ to enable rapid decay of $a_{ij}\theta_{ij}(k)$ for all $k\geq k_f$.
\end{proof}

The conditions on $\sigma$ in Proposition~\ref{prop:sigma} are actually the $\mathcal{I}_\sigma^f$ in Algorithm~1. These necessary conditions can provide better assurance for the resilient performance.
\subsection{Sufficient condition for LRUUB}
\begin{theorem}[Sufficient condition for LRUUB]\label{theorem2}
Under Assumption~\ref{as:topology} with $f>0$ and Assumption~\ref{as:misbehaving}, if the leaders are static and $\exists \sigma \in \mathcal{I}_\sigma^f \cap \mathcal{I}_\sigma^p$, the system~(\ref{dynamic}) with \eqref{foll_al} and \eqref{eq:aij} using parameters $\{\sigma,\ \eta\}$ obtained from Algorithm~\ref{alg} can achieve LRUUB.
\end{theorem}

Several lemmas are introduced to support the proof of Theorem~\ref{theorem2}.

Above all, according to Theorem 5.1 in \citep{BLANCHINI1995451}, the UUB condition in LRUUB can be guaranteed by the existence of a Lyapunov function outside $\mathcal{I}_U$. 
\begin{lemma} \label{UUBconditionlemma}
If there exists a Lyapunov function outside $\mathcal{I}_U$ for the system~(\ref{dynamic}), then the system~(\ref{dynamic}) is UUB in $\mathcal{I}_U$.
\end{lemma}
In particular, a Lyapunov function outside $\mathcal{I}_U$ for the system~(\ref{eq:error1}) can be defined as a gauge function $\Psi:\ \mathbb{R}^{N_F} \rightarrow \mathbb{R}_+$. We briefly present some background material necessary for $\Psi$. Following \citep{272351}, we call a function $\Psi$ a gauge function if $\Psi(x) \geq 0$, $\Psi(x)=0 \Leftrightarrow x=0$; for $\mu>0, \Psi(\mu x)=\mu \Psi(x)$; and $\Psi(x+y) \leq \Psi(x)+\Psi(y), \forall x, y \in \mathbb{R}^{N_F}$. A gauge function is convex and it defines a distance of $x$ from the origin which is linear in any direction. If $\Psi$ is a gauge function, we define the closed set (possibly empty) $\bar{N}[\Psi, \epsilon]=$ $\left\{x \in \mathbb{R}^{N_F} | \Psi(x) \leq \epsilon\right\}$. 
%It is easy to show that the set $\bar{N}[\Psi, \epsilon]$ is a C-set for all $\epsilon>0$. On the other hand, any $\mathrm{C}$-set $\mathcal{S}$ induces a gauge function $\Psi_{\mathcal{S}}(x)$ (Known as Minkowski function of $\mathcal{S})$, which is defined as $\Psi(x) \doteq \inf \{\mu>0: x \in \mu \mathcal{S}\}$. Therefore a C-set $\mathcal{S}$ can be thought of as the unit ball $\mathcal{S}=\bar{N}[\Psi, 1]$ of a gauge function $\Psi$ and $x \in \mathcal{S} \Leftrightarrow \Psi(x) \leq 1$.
In particular, $\Psi$ holds the following conditions:
\begin{itemize}
\item If $ \mathbf{e}(k) \notin \bar{\mathcal{N}}_k(\Psi, \epsilon)$, then 
\begin{equation}
\Psi( \mathbf{e}(k+1))-\Psi( \mathbf{e}(k)) <0
\label{uub1}
\end{equation} 
\item If $ \mathbf{e}(k) \in \bar{\mathcal{N}}_k(\Psi, \epsilon)$, then
\begin{equation}
\Psi( \mathbf{e}(k+1)) \leq \epsilon
\label{uub2}
\end{equation} 
\end{itemize}

In our research, let $\Psi( \mathbf{e}(k))=\Vert \mathbf{e}(k) \Vert$ and $\bar{\mathcal{N}}_k(\Psi, \epsilon)=\{ \mathbf{e}(k) \in \mathbb{R}^{N_F}| \Psi( \mathbf{e}(k)) \in \mathcal{I}_U \}$. 

Next, according to Proposition 4.1 in \citep{CSFRL}, the following lemma can be derived directly.

\begin{lemma} 
Under Assumption~\ref{as:topology} with $f>0$ and \ref{as:misbehaving}, for $\sigma \in \mathcal{I}_\sigma^p$, there exists a $\Xi(k)\in (0,1)$, such that \[ \Vert A(k+h-1)\cdots A(k+1) A(k) \mathbf{e}(k) \Vert \leq (1-\Xi(k)) \Vert \mathbf{e}(k)\Vert\] where
$
\Xi(k)= ( g \cdot \underline{a_{l}}(k) )^{h}$ and $h$ is the depth of the graph. Here, $\underline{a_{l}}(k)=\underline{q_{l}}(k)/\overline{q_i^s}$; $\underline{q_{l}}(k)= \min_{k \leq \kappa \leq k+h-1}q_{il}(\kappa)$; $q_{il}(k)=\exp(-\sigma r_{il}(k))$; and $r_{il}(k)$ is the evaluation function of $e_i(k)$, i.e., $r_{il}(k)=r_{il}(k-1)+\eta [e_i(k)^2-r_{il}(k-1)]$, $r_{il}(0)=e_i(0)$ for all $i \in \mathcal{V}_F$.
\label{xilemma}
\end{lemma}

% \begin{lemma}\label{mathlemma1}
% Given a constant $c\geq 1$, for all $x \in [0,1)$,
% \[(1-x)^{\frac{1}{c}} \leq 1-\frac{1}{c}x.\] 
% The equality holds when $x=0$.
% \end{lemma}

\begin{lemma} \label{mathlemma2}
Define a continuous function $F_c: [0, c_m] \rightarrow \mathbb{R}$
\begin{equation} \nonumber
F_c(x)=-c_1 \exp(-c_3 \sigma x^2)x+c_2 \exp(-c_4\sigma(c_m-x)^2)(c_m-x),
\end{equation}
where $c_1,c_2,c_3,c_4,c_m,c_n$ are positive constants, $c_m> c_n$ and $\sigma \in \mathcal{I}_\sigma^p$. If the following two conditions hold,
\begin{itemize}
 \item 
$
c_m-(2 c_4\sigma)^{-\frac{1}{2}}\geq \max \{ (2 c_3\sigma)^{-\frac{1}{2}}, c_n\}, 
$
\item there exists two real number $x_1$ and $x_2$ satisfying $0<x_1<x_2<c_n$, such that $F_c(x_1)\leq 0$ and $F_c(x_2)\leq 0$,
\end{itemize}
then $F_c(x) \leq 0$ for all $x \in [x_1, x_2]$.
\end{lemma}
The proofs of Lemma~\ref{xilemma} and \ref{mathlemma2} are in the Appendix.

Now proceed with the proof of Theorem~\ref{theorem2}.
\setcounter{pot}{1}
\begin{pot}
As long as the system~(\ref{eq:error1}) achieves LRUUB, the system~(\ref{dynamic}) with \eqref{foll_al} and \eqref{eq:aij} also ahieves LRUUB. Referring to Definition~\ref{RAC} and Lemma~\ref{UUBconditionlemma}, the proof of LRUUB is dissembled into four parts.

\noindent $\bullet$ Firstly, we derive the one-step recursive formula of \eqref{eq:error1}. 

Under Assumption~\ref{as:topology}, there always exist at least $g$ paths from $g$ normal leaders to any normal follower. In this case, according to the Lemma~\ref{xilemma} and Proposition 4.1 in \citep{CSFRL}, since $a_{ij}$ is bounded, for all $k\in \mathbb{N}$ and $k'\in \mathbb{N}_0$, we have 
$$\Vert A(k+k')A(k+k'-1)\cdots A(k)\mathbf{e}(k) \Vert \leq (1-\Xi(k))^{\frac{k'}{h}} \Vert \mathbf{e}(k) \Vert .$$ 
According to Bernoulli's inequality, for $\Xi(k)\in (0,1)$, we have $$1-(1-\Xi(k))^{\frac{1}{h}}>\frac{\Xi(k)}{h}.$$
Based on norm properties, we have:
\begin{equation}\nonumber
\begin{split}
&\Vert \mathbf{e}(k+1) \Vert \leq \Vert A(k) \mathbf{e}(k) \Vert + \Vert L_m(k)\mathbf{e}^m(k)-L_{q}(k) \mathbf{e}(k)\Vert\\
&\leq (1-\Xi(k))^{\frac{1}{h}} \Vert \mathbf{e}(k) \Vert + \Vert L_m(k)\mathbf{e}^m(k)-L_{q}(k) \mathbf{e}(k)\Vert\\
&\leq (1-\frac{\Xi(k)}{h})\Vert \mathbf{e}(k) \Vert + \Vert L_m(k)\mathbf{e}^m(k)-L_{q}(k) \mathbf{e}(k)\Vert.
\end{split}
\end{equation}

Define $\theta_{im}(k)=e_i(k)-c_m$, $r_{im}(k)= (1-\eta)r_{im}(k-1)+\eta \theta_{im}(k)$ and $q_{im}(k)=\exp(-\sigma r_{im}(k))$. Let $\overline{q}_{im}(k)=\max_{i\in \mathcal{V}_F} q_{im}(k)$. In these variables, we do not distinguish the labels of misbehaving agents. Under Assumption~\ref{as:topology}, for each normal agent, there exist at least $g$ normal neighbors and at most $f$ misbehaving neighbors. We estimate that 
\begin{equation} \label{thpr0}
\begin{split}
\Vert \mathbf{e}(k+1)\Vert-\Vert \mathbf{e}(k)\Vert & \leq -\frac{1}{h}(g\overline{q_i^s}^{-1}(k)\underline{q_l}(k))^{h}\Vert \mathbf{e}(k)\Vert\\
& + f \overline{q_i^s}^{-1}(k) \overline{q}_{im}(k) (c_m-\Vert \mathbf{e}(k) \Vert) 
\end{split}
\end{equation}

\noindent $\bullet$ Secondly, we will prove the safety condition in Definition~\ref{RAC}, that is $\Vert \mathbf{e}(k)\Vert \leq c_n$ for all $k \in \mathbb{N}$. 

We first illustrate $\Vert \mathbf{e}(1) \Vert \leq \Vert \mathbf{e}(0) \Vert$, such that $\Vert \mathbf{e}(1) \Vert \leq c_n$. In Algorithm~1, $F_m(c_n)\sigma \geq F_t(c_n)$ always holds and thus, we have $\exp(F_m(c_n)\sigma) \geq \exp(F_t(c_n))$. Substituting $F_m$ and $F_t$ into it, we obtain
\begin{equation}
\begin{split}
 &-\frac{1}{h}[\overline{q_i^s}^{-1} g \exp(-\sigma c_n^2)]^{h} c_n \\
 &+ \overline{q_i^s}^{-1} f \exp(-\sigma(c_m-c_n)^2)(c_m-c_n) \leq 0. 
\end{split}
\label{thpr1}
\end{equation} 

Under Assumption~\ref{as:misbehaving}, we have $\Vert \mathbf{e}^m(0) \Vert \geq c_m$ and $\Vert \mathbf{e}(0) \Vert \leq c_n$. Substituting (\ref{thpr1}) into (\ref{thpr0}), there is $\Vert \mathbf{e}(1) \Vert \leq \Vert \mathbf{e}(0) \Vert $. 

Then we prove the safety condition.
Suppose, by contradiction, there exists a time-step $k^*$ such that 
$\Vert \mathbf{e}(k^*)\Vert =e^*>c_n$ while $\Vert \mathbf{e}( k^*-1)\Vert \leq c_n$. According to (\ref{thpr1}), we have 
\begin{equation}\nonumber
\begin{aligned}
 a_{im}(0) \theta_{im}(0) &\leq \overline{q_i^s}^{-1} \exp(-\sigma(c_m-c_n)^2)(c_m-c_n)\\
 &\leq \frac{1}{h}(\overline{q_i^s}^{-1} g \exp(-\sigma c_n^2))^{h} c_n = \frac{\Xi(0)}{h}c_n
\end{aligned}
\end{equation} 
Applying the recursion of \eqref{thpr0} to $k^*$ time step, we deduce that
\begin{equation}\nonumber
\begin{aligned}
e^* =&\Vert \mathbf{e}(k^*) \Vert \leq \Vert A(k^*-1) \Vert\ \Vert \mathbf{e}(k^*-1)\Vert +
 \Vert L_m(k^*-1) \mathbf{e}^m(k^*-1)\\
&\qquad -L_{q}(k^*-1) \mathbf{e}(k^*-1)\Vert\\
&\leq \Vert A(k^*-1) \cdots A(0) \Vert\ \Vert \mathbf{e}(0)\Vert + \Vert L_m(k^*-1) \mathbf{e}^m(k^*-1)\\
&\qquad -L_{q}(k^*-1) \mathbf{e}(k^*-1)\Vert+ \sum_{\kappa=2}^{k^*} \Vert \prod \limits_{c=1}^{\kappa-1} A(k^*-c)\Vert\\
& \qquad \Vert L_m(k^*-\kappa)\mathbf{e}^m(k^*-\kappa)-L_{q}(k^*-\kappa) \mathbf{e}(k^*-\kappa)\Vert\\
\end{aligned}
\end{equation}
 \begin{equation}\nonumber
 \begin{aligned}
 &\leq (1-\Xi(0))^{\frac{k^*}{h}} \Vert \mathbf{e}(0)\Vert + \sum_{\kappa=1}^{k^*} (1-\Xi(0))^{\frac{\kappa-1}{h}} \Vert L_m(k^*-\kappa)\\
&\qquad \mathbf{e}^m(k^*-\kappa)-L_{q}(k^*-\kappa) \mathbf{e}(k^*-\kappa)\Vert\\
&\leq(1-\Xi(0))^{\frac{k^*}{h}}\Vert \mathbf{e}(0) \Vert+\frac{1-(1-\Xi(0))^{\frac{k^*}{h}}}{1-(1-\Xi(0))^{\frac{1}{h}}} f a_{im}(0) \theta_{im}(0)\\
&\leq [(1-\Xi(0))^{\frac{k^*}{h}}+\frac{1-(1-\Xi(0))^{\frac{k^*}{h}}}{1-(1-\Xi(0))^{\frac{1}{h}}}\frac{\Xi(0)}{h}]c_n
 \end{aligned}
 \end{equation}

Based on Bernoulli's inequality, we have $e^*<c_n$, which is contrary to our assumption and thus, the safety condition is satisfied.

\noindent $\bullet$ Thirdly, we prove the condition (\ref{uub1}) when $\mathbf{e} \notin \mathcal{I}_U$.

Let $\underline{r_{im}}(k)=\min_{i\in \mathcal{V}_F} r_{im}(k)$, $\overline{r}_{il}(k)=\max_{i \in \mathcal{V}_F} r_{il}(k)$. For any $i \in \mathcal{V}_F$, we have
\begin{eqnarray*}
&r_{im}\geq (1-\eta)\underline{r_{im}}(k-1)+ \eta(c_m- \Vert \mathbf{e}(k)\Vert)^2\\
&r_{il}\leq (1-\eta)\overline{r}_{il}(k-1)+\eta \Vert \mathbf{e}(k)\Vert^2.
\end{eqnarray*}
According to the one-step recurrence formula (\ref{thpr0}), we have 

\begin{equation}
\begin{split}
&\Psi( \mathbf{e}(k+1))-\Psi( \mathbf{e}(k)) \leq \\
&-\frac{1}{h} \Vert \mathbf{e}(k)\Vert [\overline{q_i^s}^{-1} g \exp[-\sigma((1-\eta)\overline{r}_{il}(k-1)+\eta \Vert \mathbf{e}(k)\Vert^2)]]^{h}\\
& + \overline{q_i^s}^{-1} f \exp[-\sigma( (1-\eta)\underline{r_{im}}(k-1)+ \eta(c_m- \Vert \mathbf{e}(k)\Vert)^2)] \\
& (c_m- \Vert \mathbf{e}(k)\Vert) \triangleq F_D(\Vert \mathbf{e}(k)\Vert, \overline{r}_{il}(k-1),\underline{r_{im}}(k-1)).\\
\end{split}
\label{FD}
\end{equation}

\noindent The scale result is defined as a function $F_D: [ \epsilon, c_n] \times \mathbb{R}_+ \times \mathbb{R}_+ \rightarrow \mathbb{R}$.

Since $\Vert \mathbf{e}(k)\Vert \leq c_n$ for all $k$, we have $\overline{r_{il}}(k)\leq c_n^2$ and $\underline{r_{im}}(k) \geq (c_m-c_n)^2$ for all $k$. Define 
\begin{equation}\nonumber
F_\Psi( \Vert \mathbf{e}(k)\Vert)\triangleq F_D\vert _{\substack{\overline{r}_{il}(k-1)=c_n^2},\\ \underline{r_{im}}(k-1)=(c_m-c_n)^2}
\end{equation}
which is a $F_c$-like function in Lemma~\ref{mathlemma2} with the parameters $c_3=\eta h$, $c_4=\eta$. 

Under $\sigma \geq \underline{\sigma}_p$, for $ \sigma \geq \frac{1}{2(c_m-c_n)^2}$, there always exists $\eta \in [\frac{1}{2\sigma (c_m-c_n)^2}, 1]$, such that $c_m-(2 c_4\sigma)^{-\frac{1}{2}} \geq c_n$; For $ \sigma \geq \frac{\sqrt{6}+3}{2 c_m^2} \geq \frac{(1+\sqrt{\frac{1}{h}})^2}{2 c_m^2}$, there always exists $\eta \in [\frac{(1+\sqrt{h})^2}{2\sigma h c_m^2}, 1]$, such that $c_m-(2 c_4\sigma)^{-\frac{1}{2}} \geq (2 c_3\sigma)^{-\frac{1}{2}}$. Overall, for $\sigma \geq \underline{\sigma}_p$, there always exists $\eta \geq \max \{ \frac{(1+\sqrt{h})^2}{2\sigma h c_m^2}, \frac{1}{2\sigma (c_m-c_n)^2}\}$, such that
$c_m-(2 c_4\sigma)^{-\frac{1}{2}}\geq \max \{ (2 c_3\sigma)^{-\frac{1}{2}}, c_n\}$. Similarly, as if there exists $\sigma$ satisfying $F_m(\epsilon)\sigma \geq F_t(\epsilon)$, it can guarantee the existence of $\eta \in [\frac{\sigma^{-1}F_t(\epsilon)-F_m(c_n)}{(c_n-\epsilon)[2c_m+(h-1)(c_n+\epsilon)]},1]$, such that 

\begin{equation}
\begin{split}
&\ln[ \frac{1}{h} \epsilon (\overline{q_i^s}^{-1} g \exp[-\sigma((1-\eta)c_n^2+\eta \epsilon ^2)])^{h}]\geq \ln [ (c_m-\epsilon)\\ 
& \overline{q_i^s}^{-1} f \exp[-\sigma( (1-\eta) (c_m-c_n)^2+\eta(c_m- \epsilon)^2)]]. 
\end{split}
\label{thpr2}
\end{equation}

Hence, we have $F_\Psi(\epsilon)\leq 0$. In addition, $F_\Psi(c_n)\leq 0$ has been illustrated by (\ref{thpr1}).
Due to Lemma~\ref{mathlemma2}, it can be asserted that $F_\Psi(\Vert \mathbf{e}(k)\Vert) \leq 0$ for all $\Psi( \mathbf{e}(k)) \in (\epsilon, c_n]$. Then (\ref{uub1}) is established.
 
\noindent $\bullet$ Fourthly, for $\Vert \mathbf{e}(k)\Vert \leq \epsilon$, similar to the first part, it is not hard to ensure the condition (\ref{uub2}) and thus the proof is omitted. The only difference is that let $\Vert \mathbf{e} \Vert \leq \epsilon$ instead of $\Vert \mathbf{e} \Vert\leq c_n$. On the boundary of $\mathcal{I}_U$, the difference of $\Psi( \mathbf{e})$ is always negative under (\ref{thpr2}) and thus, (\ref{uub2}) is established.

Theorem~2 is established.\hfill $\blacksquare$
\end{pot}
\begin{remark}
In the foregoing proof, the inequality (\ref{FD}) is fairly conservative, because we utilize $r_{im}(k-1) \geq (c_m-c_n)^2$ and $q_i^s(k) \leq \overline{q_i^s}$ when $\Vert \mathbf{e}(k)\Vert=\epsilon$. Actually, the ultimate bound of $ \mathbf{e}$ is much less than $\epsilon$. For the system~(\ref{eq:error1}), $\Vert \mathbf{e} \Vert$ can reach $\mathcal{I}_U$ in finite time steps $k_f$, where
\begin{equation} \label{finite_time}
k_f <\frac{c_n-\epsilon}{\min \{F_\Psi(c_n), F_\Psi(\epsilon) \}}.
\end{equation}
For $k\geq k_f$, we can estimate that $ q_i^s(k) \leq d_{max}-f$ by $q_{im}(k) \ll q_{ij}(k)$. So $\Xi(k) \in [(\frac{g}{d_{max}-f})^h,1)$. In addition, for $k\geq k_f$, $A(k)$ tends to stabilize at a constant matrix $A_f$ satisfying $||A_f|| < 1$. It follows from (\ref{eq:error1}) that $|| \mathbf{e}(k)|| \leq ||A_f||^{k-k_f} ||\mathbf{e}(k_f)|| +b$ and thus, the actual ultimate bound is the remainder term $b$. By recursion, we have 

\begin{equation}
\begin{split}
b &\leq (q_i^s)^{-1} f F_i(c_m-\epsilon) \sum_{k=0}^\infty (1-\Xi(k))^{\frac{k}{h}}\\
&\leq (q_i^s)^{-1} f F_i(c_m-\epsilon) h \sum_{k=0}^\infty (1-\Xi(k))^{k}\\
&\leq (q_i^s)^{-1} f F_i(c_m-\epsilon) h \Xi^{-1}(k) \\
&\leq g^{-h}(d_{max}-f)^{h-1}f h F_i(c_m-\epsilon).
\end{split}
\label{B}
\end{equation}

\end{remark}
\noindent Generally, the actual ultimate bound $b \ll \epsilon$.

For better controller performance and wider application range, two corollaries are proposed.
The first corollary gives the LRUUB condition for the fastest restoration of the system when misbehaving events occur after the system converges. 
\begin{corollary}\label{co1}
Consider that a misbehaving event occurs at $k_0$ time step instead of occurring at $0$ time step, where $k_0 \geq k_f$ in (\ref{finite_time}). Under Assumption~\ref{as:topology} with $f>0$ and \ref{as:misbehaving}, if  parameters $\{ \sigma, \eta \}$ obtained from Algorithm~1 satisfy
\begin{equation}\label{co1eq1}
\eta \geq \frac{ \ln ( f (c_m-\epsilon))-\ln (\overline{q_i^s} c_n)}{\sigma (c_m-\epsilon)^2},\ \text{and}
\end{equation}
\begin{equation}\label{co1eq2}
\eta \geq \frac{(c_m-c_n)^2-h(c_n^2-\epsilon^2)}{(c_m-\epsilon)^2},
\end{equation}
the system~(\ref{dynamic}) with \eqref{foll_al} and \eqref{eq:aij} achieve LRUUB and normal followers restore at only $k_0+1$ time step. 
\end{corollary}
\begin{proof}
If the misbehaving event occurs at $k_0$ time step, the system is disturbed by $r_{im}(k_0)\geq (1-\eta)\epsilon^2+\eta(c_m-\epsilon)^2$. Assume that $\Vert \mathbf{e}(k_0+1)\Vert =e^*$ estimated by
$e^*\leq | f a_{im}(k_0)\theta_{im}(k_0)| \leq \overline{q_i^s}^{-1} f \exp(-\sigma \eta (c_m-\epsilon)^2)(c_m-\epsilon)$. 
%According to (\ref{thpr2}), we have 
%$e^* \leq \exp[-\sigma(1-\eta)(c_m-c_n)^2] \frac{1}{h} \epsilon (\overline{q_i^s}^{-1} g \exp[-\sigma((1-\eta)c_n^2+\eta \epsilon ^2)])^{h}$.
Passing (\ref{co1eq1}) into it, $e^*\leq c_n$ is established and the safety condition is satisfied. 

Next, we will prove the condition~\eqref{uub1} and illustrate the restoration of normal followers, which is ensured by 
\begin{equation}\nonumber
\Psi( \mathbf{e}(k+1))-\Psi( \mathbf{e}(k)) \leq F_D(e^*, \epsilon^2, \eta(c_m-\epsilon)^2)\leq 0.
\end{equation} 
for $k\geq k_0+1$. The following inequality
\begin{equation}\label{eq:co1}
\begin{split}
&\frac{\exp[-\sigma h((1-\eta)\epsilon^2+\eta x^{*2})]}{\exp[-\sigma((1-\eta)\eta(c_m-\epsilon)^2+\eta(c_m-e^*)^2)]}\geq \\
&\frac{\exp[-\sigma h((1-\eta)c_n^2+\eta x^{*2})]}{\exp[-\sigma((1-\eta)(c_m-c_n)^2+\eta(c_m-e^*)^2)]}. 
\end{split}
\end{equation}
can be ensured by (\ref{co1eq2}).
Applying it into (\ref{FD}), we can directly obtain $\Psi( \mathbf{e}(k_0+2))-\Psi( \mathbf{e}(k_0+1)) \leq F_\Psi(e^*)\leq 0$. Then $k_0+1$ can be regarded as the initial time step and the initial states satisfy $\Vert \mathbf{e}(k_0+1) \Vert \leq c_n$. According to Theorem~\ref{theorem2}, the system can achieve LRUUB.
\end{proof}

According to Lemma~\ref{xilemma}, the agent closer to the leader generally converges faster. In this case, the second corollary gives a more general topological condition to relax Assumption~\ref{as:topology}.

% Replace the certain edges in $\mathcal{E}'=\{(j,i) \in \mathcal{E}| i\in\mathcal{V}^{(n)},j \in \mathcal{V}^{(n-1)},\ n=2,\cdots,h\} $ by the same amount new edges in $\mathcal{E}"=\{ (p,i)| i\in\mathcal{V}^{(n)}, p \in \mathcal{V}^{(n')},\ n=2,\cdots,h,\ n'>n-1\} $.
\begin{corollary}\label{coro}
Under Assumption~\ref{as:misbehaving}, consider a networks of agents in $\mathcal{G}=\{\mathcal{V},\ \mathcal{E}\}$, of which some agents $i$ in $\mathcal{V}^{(n)}$ have less than $g+f$ neighbors in $\mathcal{V}^{(n-1)}$. Let $\mathcal{E}'=\{(j,i)| (j,i) \in \mathcal{E},i\in\mathcal{V}^{(n)},j \in \mathcal{V}^{(n-1)},\ n=2,\cdots,h\} $ and $\mathcal{E}"=\{ (p,i)| (p,i)\notin \mathcal{E}, i\in\mathcal{V}^{(n)}, p \in \mathcal{V}^{(n')},\ n=2,\cdots,h,\ n'>n-1\}$, which have the same number of entries. If the leaders are static and $\mathcal{G}'=\{\mathcal{V},\ \mathcal{E}\setminus \mathcal{E}' \cup \mathcal{E}" \}$ satisfies Assumption~\ref{as:topology}, the origin in $\mathcal{G}$ can achieve LRUUB with parameters $\{ \sigma,\eta\}$ designed by Algorithm~1 using input parameters of $\mathcal{G}'$.
\end{corollary}
\begin{proof}
According to Lemma~\ref{xilemma}, since $\mathcal{G}'$ satisfies Assumption~\ref{as:topology}, the convergence rate of any $p \in \mathcal{V}^{(n')}$ in the graph $\mathcal{G}'$ can be scaled by $ |e_p(k+1)|-|e_p(k)|\leq -\frac{1}{n'} (g \underline{a_{l}}(k))^{n'} |e_p(k)|+\overline{q_i^s}^{-1} f \exp(-\sigma(c_m-c_n)^2)(c_m-c_n)$. According to Theorem~\ref{theorem2}, the system in $\mathcal{G}'$ with parameters $\{\sigma,\eta\}$ designed by Algorithm~1 can achieve LRUUB and $|e_p(k+1)|-|e_p(k)|<0$ for $e_p(k) \notin \mathcal{I}_U$. Since $\frac{\partial |e_p(k+1)|-|e_p(k)|}{\partial n'} >0$ for all $n'>1$, the agent closer to the leader converges faster. For $j\in \mathcal{V}^{(n-1)}$ and $p \in \mathcal{V}^{(n')}, n'>n-1$, we have $|e_j(k+1)|-|e_j(k)| < |e_p(k+1)|-|e_p(k)|< 0$. Let $e_i'$ be the error in the case where $i$ with $(p,i) \in \mathcal{E}\setminus \mathcal{E}' \cup \mathcal{E}"$ in $\mathcal{G}'$, and $e_i$ be the error in the case where $(j,i) \in \mathcal{E}$ in $\mathcal{G}$. We have $|e_i(k+2)|-|e_i(k+1)|\leq |e_i'(k+2)|-|e_i'(k+1)|<0$. A faster convergence neighbor prompts faster convergence of agent $i$. Hence, as if the system in $\mathcal{G}'$ can achieve LRUUB, the origin in $\mathcal{G}$ can achieve LRUUB as well.
\end{proof}
% According to (\ref{FD}), a faster convergence of $i$ can also be proved with misbehaving agents.

\section{Simulations}
We conduct simulations in MATLAB/Simulink to demonstrate the effectiveness and resilience of our proposed protocol in the MG system shown in Fig.~\ref{fig:electrical}. The electrical parameters and loads are the same as in \citep{IEEE33}, with the 10-11 bus link changed to the 4-11 bus link. After opening Breaker 1, the IEEE 33-bus system transitions to islanded operation mode. The droop coefficients of DESs in $\{19,11,7,20,13,8,21,15,9,17,10 \}$ are set as $5\cdot 10^{-4}\ rad/(s \cdot kW)$ and those in $\{28,30,32,22,23,24\}$ are set as $1\cdot 10^{-3}\ rad/(s \cdot kW)$. In a ($f=3$)-local misbehaving network, misbehaving events may occur in the $5$-robust ($g=2$) digraph $\mathcal{G}=\{\mathcal{V}, \ \mathcal{E} \}$ in Fig.~\ref{fig:communication}. The tertiary controller transmits the reference values in 5 channels to each DES in $V^{1}$. The sampling time of the secondary control is set to $t_s=0.01s$.

We first design parameters according to Algorithm~1. The system parameters are defined as $c_n=\pi\ rad/s$, $c_m=3.1\pi\ rad/s$, and $\epsilon=0.001\ rad/s$ according to Chinese national standards. The basic scope of $\sigma$ is $\mathcal{I}_\sigma^p=[0.0287, \infty]$. We calculate $\mathcal{I}_\sigma^f=[0.8911,\infty)$. Setting $\sigma=0.9$, we have $\mathcal{I}_\eta=[0.0427, 1]$ and set $\eta=0.4$. Then we conduct the following simulation scenarios:

1) $0s \leq t< 2s$: All DESs are intact and start with initial frequencies within the safety interval. Breaker 1 is open and Breaker 2 is closed. The reference values are set to $w^{l}=100 \pi\ rad/s$ and $P^{l}=0.1852\ rad/s$.

2) $2s \leq t< 4s$: The first class of misbehaving events occur on DESs 19, 20, and 21. The second class of misbehaving events occur on DESs 22, 23 and 24, and on at most 2 channels of the tertiary controller to each DES in $V^{1}$. In this stage, the network is a 2-local misbehaving network. The $w^l$ is set to $99.6 \pi\ rad/s$ and $P^{l}=0.24755\ rad/s$. 

3) $4s \leq t< 6s$: Breaker 2 is open, such that the loads on bus 19, 20, 21 are cut off from the MG. In addition to the misbehaving events in the previous stage, the second class of misbehaving events also occurs on DESs 28 and 32, making the network a 3-local misbehaving network. The $w^l$ is set to $100.4 \pi\ rad/s$ and $P^{l}=0.2249\ rad/s$. 

\begin{figure}[!t]
\centering
\includegraphics[width=2.8in]{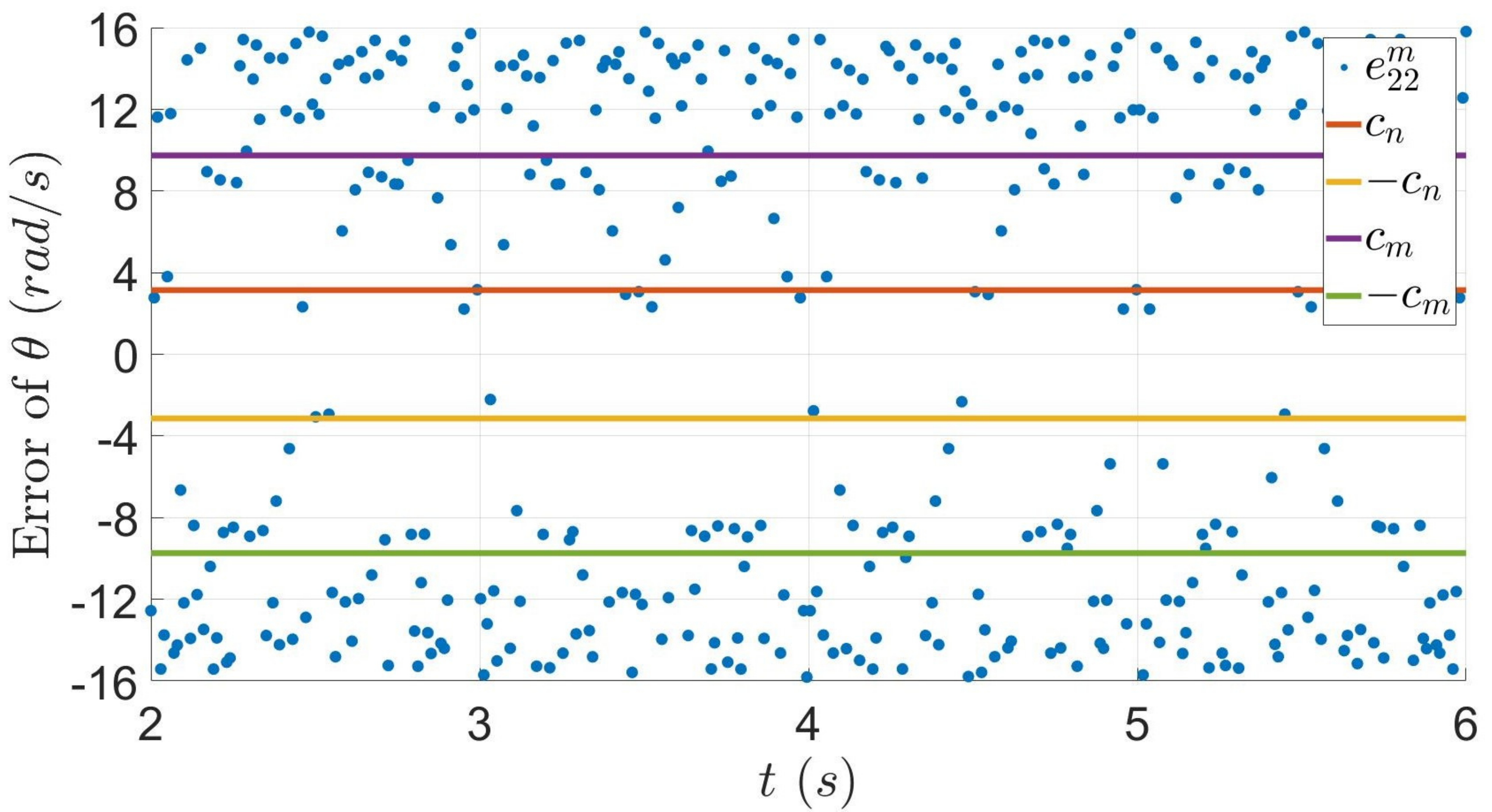}\\
\caption{The misbehaving errors of DES 22.}
\label{fig:noise}
\end{figure}
As shown in Fig.~\ref{fig:noise}, we take DES 22 as an example to describe the misbehaving errors. Suppose the square of misbehaving errors can be represented by a superposition of three frequency components: $\omega_0=0$, $\omega_1=2\pi\cdot 20.5$ and $\omega_2=2\pi\cdot 75.5$. For all $i\in \mathcal{V}_M$, they satisfy

$(e_i^m(k))^2\geq |120+20\exp(j\cdot 41\pi t_s k)+11 \exp(j\cdot 151\pi t_s k)|.$

\begin{figure}[!b]
\centering
\subfigure[The angular frequency $\omega$.]
{\includegraphics[width=2.8in]{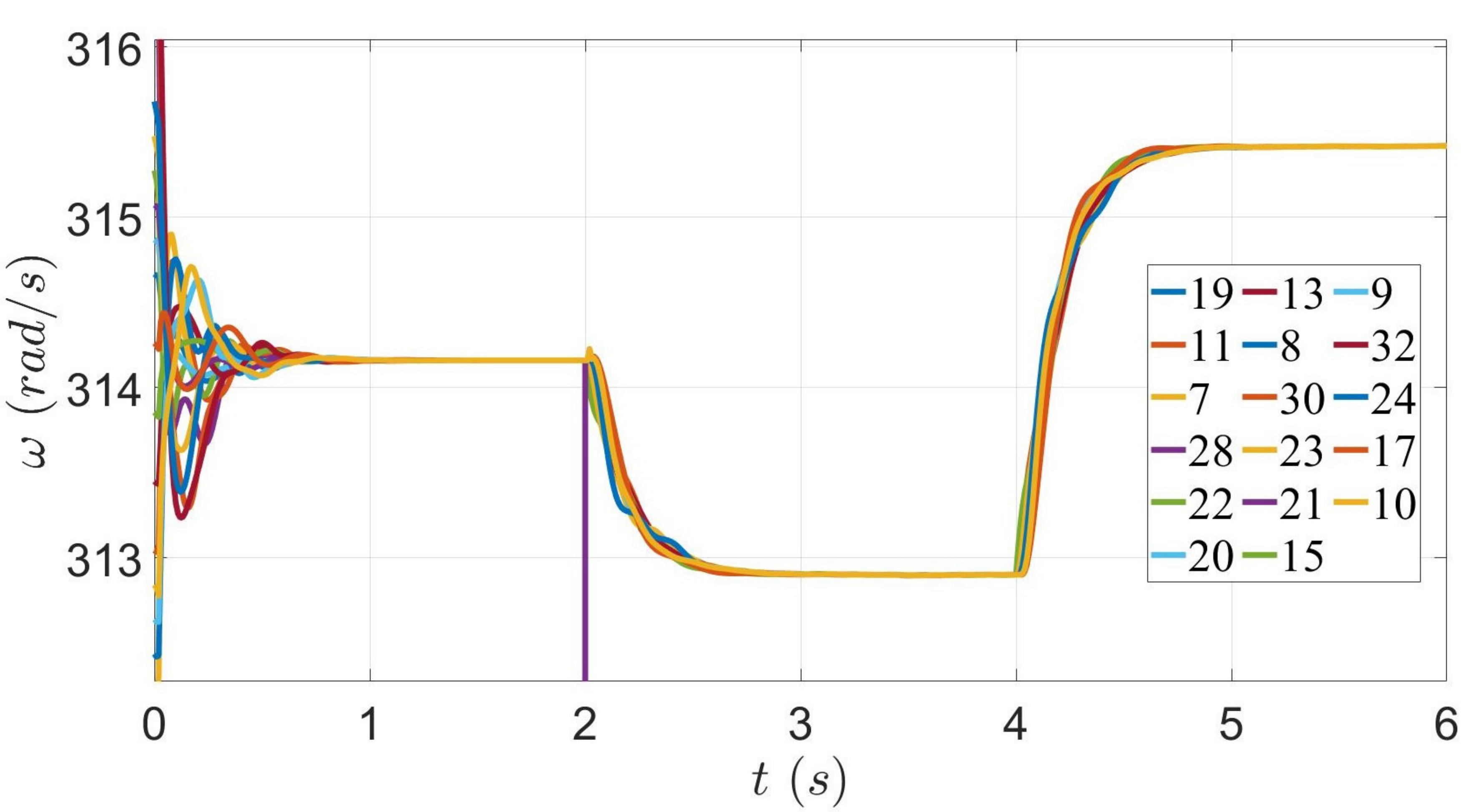}}
\subfigure[The active power $P$.]
{\includegraphics[width=2.8in]{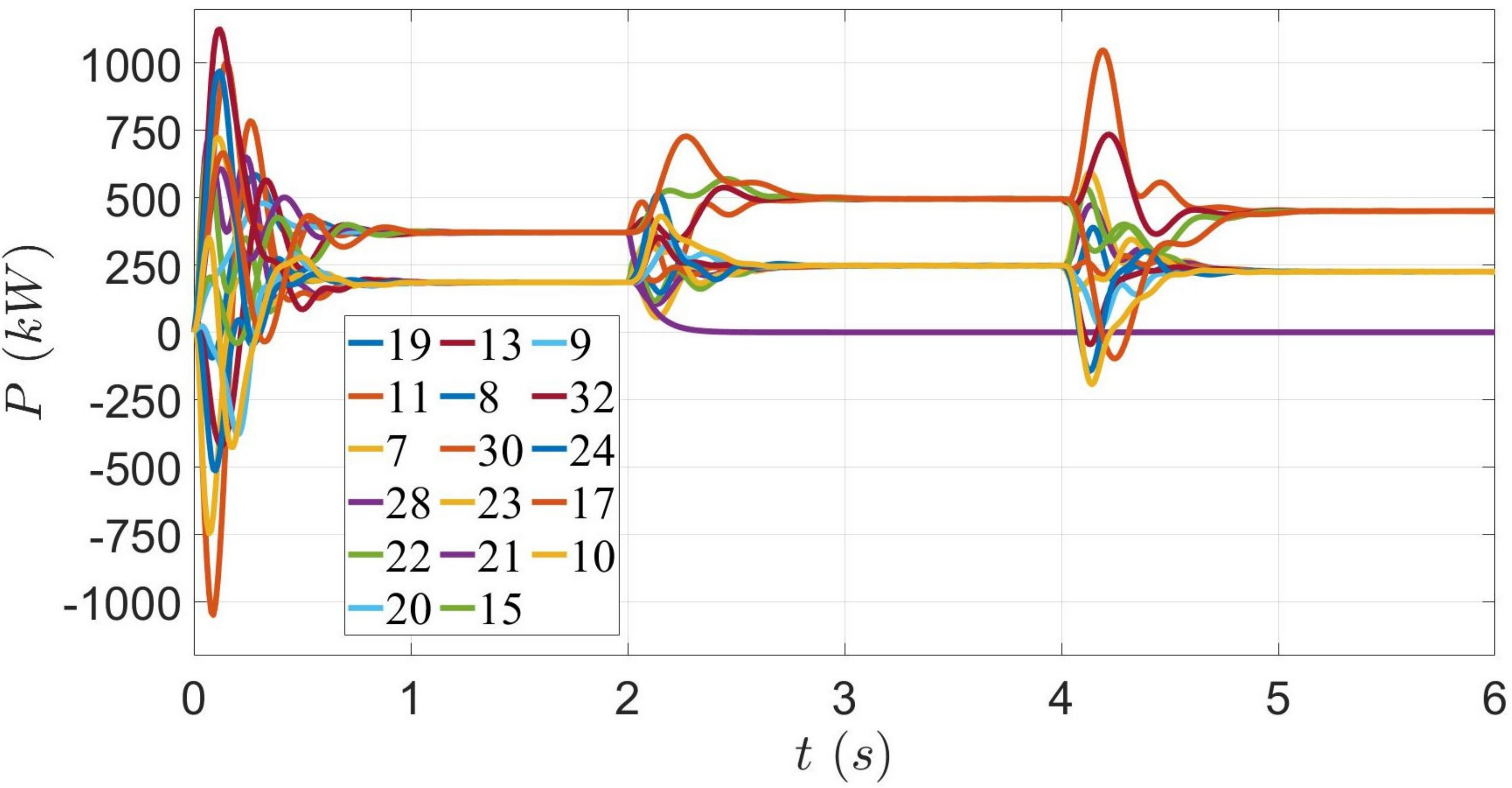}\label{fig:power performance}} 
\caption{The performance of DESs using DACC with $\{ \sigma=0.9,\ \eta=0.4 \}$.}\label{fig:DACC performance}
\end{figure}

Under $\eta=0.4$, we calculate the condition in Assumption~\ref{as:misbehaving}: $\chi_0-\sum_{\omega=\omega_0}^{K_n \omega_0}\chi_\omega[1+((\eta^{-1}-1)t_s \omega)^2]^{-\frac{1}{2}}\geq 95.4985>c_m^2$. Under Assumption~\ref{as:misbehaving}, there exists $|e_i^m(k)| < c_m$ for some $k$.

In the above simulation scenarios, the performance of DESs using DACC with $\{ \sigma=0.9,\ \eta=0.4 \}$ is depicted in Fig.~\ref{fig:DACC performance}, where the proportions of active power adhere closely to $P_{i}/P_{j}=m_{i}/m_{j}$. The results show that DACC can achieve LRUUB against the HDMA. The actual ultimate bound in \eqref{B} is calculated by $b\leq 2\times 10^{-6}$, and the errors of normal agents reach the bound of $\mathcal{I}_U$ in about $1.5s$. In contrast, if we adopt the control parameters $\{ \sigma=0.4,\ \eta=0.4 \}$, which do not satisfy Algorithm~1, the system fails to achieve UUB, as shown in Fig.~\ref{fig:diff_paras}. Similarly, if we do not adopt the discounted accumulation mechanism with $\eta=1$, which dose not satisfy Assumption~\ref{as:misbehaving}, the system cannot maintain UUB either. Hence, both Algorithm~1 and the necessary conditions must be satisfied.
\begin{figure}[!t]
\centering
\includegraphics[width=2.8in]{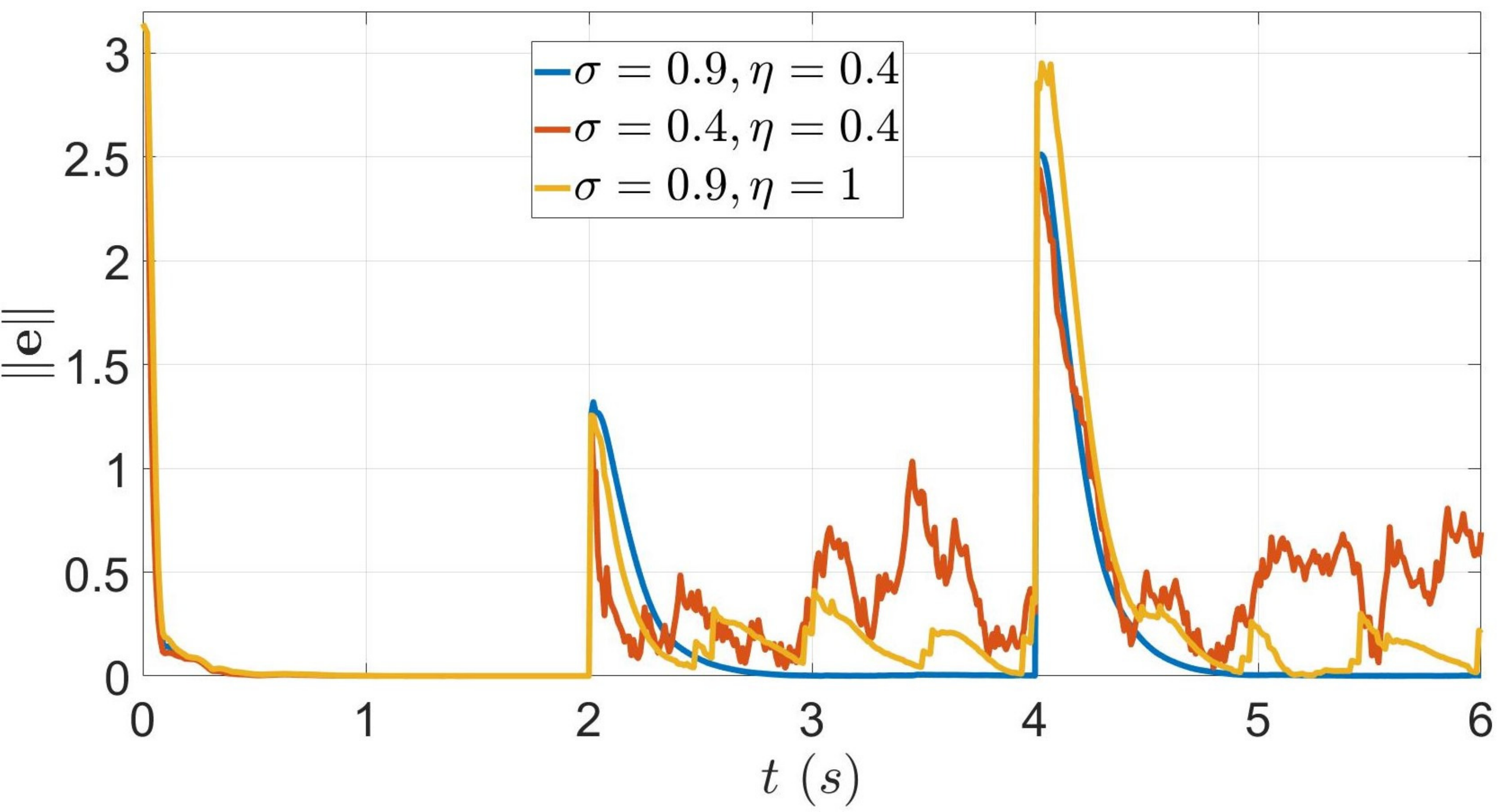}\\
\caption{Comparison results of DACC with different parameters.}
\label{fig:diff_paras}
\end{figure}

Considering the same simulation scenarios, we compare the performance of DACC with Tradition consensus (TC), W-MSR in \citep{RACR}, QW-MSR in \citep{RRQC} and Hidden layer method in \citep{FDIARD}. TC employs fixed weights $a_{ij}$ as constants, represented by the system~(\ref{dynamic}) with \eqref{foll_al} and \eqref{eq:aij} with $\sigma=0$. Both W-MSR and QW-MSR require a $(2f+1)$-robust graph, prompting a control trial in a 2-local misbehaving network. Hidden layer method requires the misbehaving states to be upper bounded and an absolutely secure observer. Note that there exists a 2-local misbehaving network in the second stage and a 3-local misbehaving network in the third stage. Due to the large errors after system divergence, a function $\ln(1+\Vert \mathbf{e} \Vert)$ is used to compare the results in this experiment. As shown in Fig.~\ref{fig:diff_methods}, DACC achieves LRUUB, while TC and Hidden layer method fail to achieve UUB in the both second and third stage. Furthermore, although W-MSR successfully achieves UUB in the second stage, it is evident that DACC outperforms MSR-type methods in the presence of HDMA. Besides, since QW-MSR is used in resilient quantized consensus problems, it may produce fluctuations in the convergence of decimal places. Hence, DACC has outstanding advantages in the context of our concerned misbehaving agents. 

\begin{figure}[!t]
\centering
\includegraphics[width=2.8in]{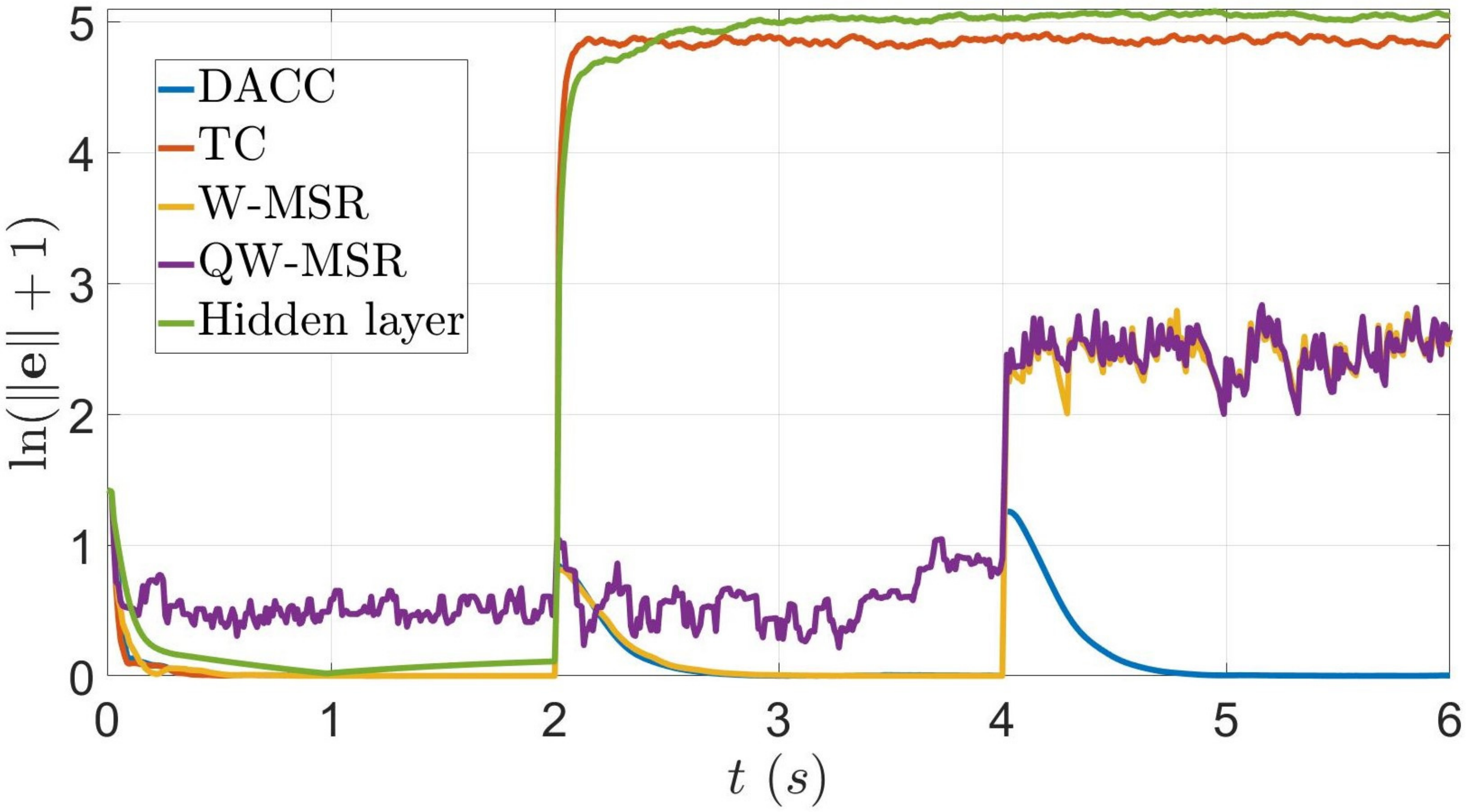}
\caption{Comparison results of ln$(1+\Vert \mathbf{e} \Vert)$ using different methods.}
\label{fig:diff_methods}
\end{figure}

To verify Corollary~\ref{co1}, whose condition is satisfied by $\eta=0.4$, we construct simulations where the misbehaving events occur at different time steps. As shown in Fig.~\ref{fig:different_time}, if misbehaving events occur during the process of the system approaching a steady state, the fluctuations generated by misbehaving agents are smaller. Especially when the system has already achieved UUB, the misbehaving events have almost no impact. The results demonstrate that DACC guarantees resilient performance regardless of when the misbehaving events occur.
% Since $\Vert B \Vert \ll \epsilon$, the states of normal followers barely change when the misbehaving event occurs. 
\begin{figure}[!t]
\centering
\includegraphics[width=2.8in]{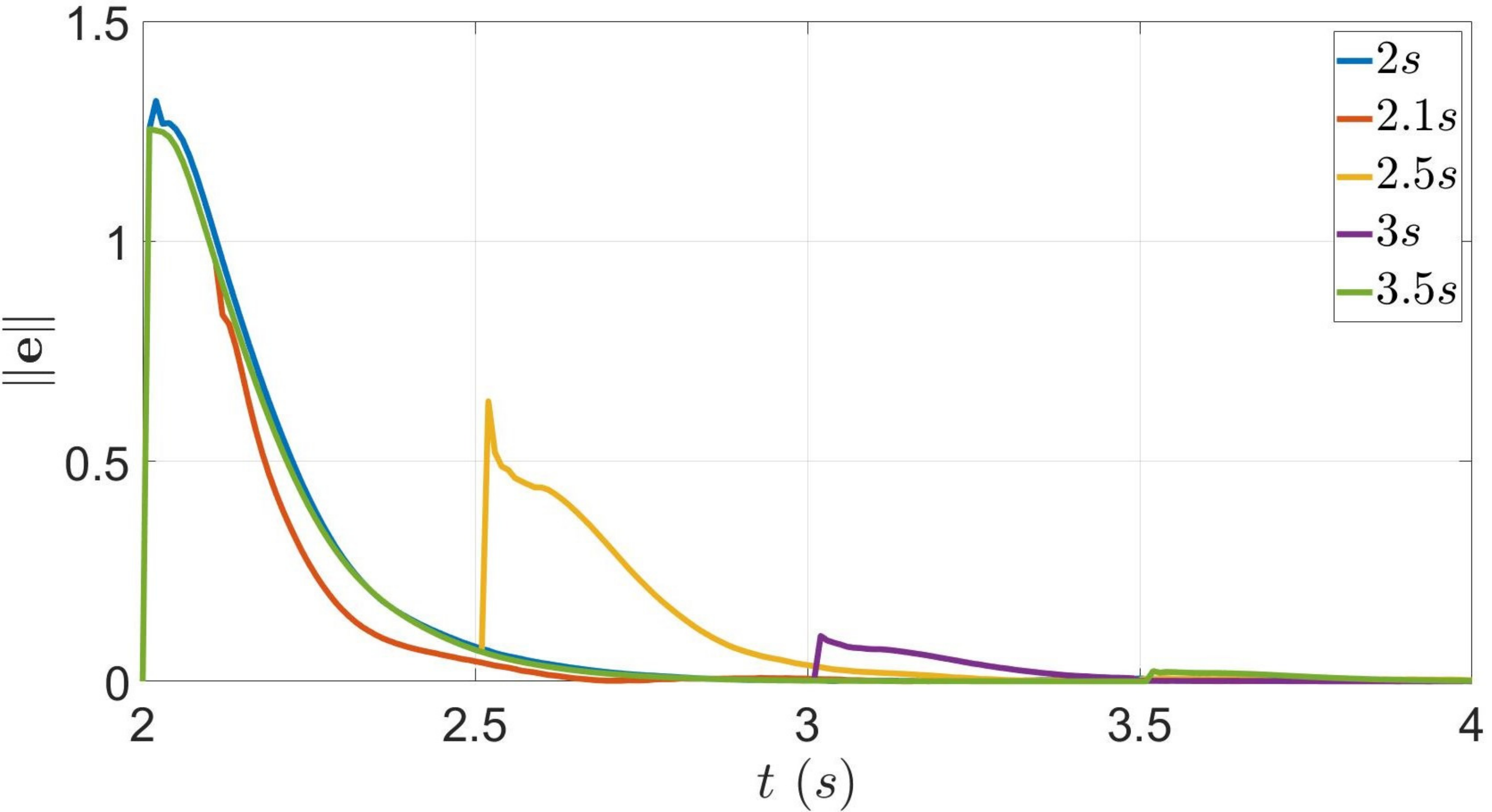}
\caption{DACC performance in the case where misbehaving events occur at different time steps in the second stage.}
\label{fig:different_time}
\end{figure}

To investigate how Corollary~\ref{coro} relaxes restrictions on the graph, we consider two cases in comparison with the original digraph $\mathcal{G}$. In Case 1, we modify $\mathcal{G}$ to form a new digraph $\mathcal{G}_1$ by replacing the edge $(13,15)$ with a new edge $(11, 15)$. In Case 2, we further modify $\mathcal{G}_1$ to create a digraph $\mathcal{G}_2$ by replacing the edge $(9,10)$ with a new edge $(7, 10)$. We maintain the same control parameters and simulation scenarios as in Fig.~\ref{fig:DACC performance} for all three cases. As shown in Fig.~\ref{difftopo}, the convergence rates are faster for agents closer to the leaders in the second stage. When designing control parameters under a graph $\mathcal{G}'$ that does not satisfy Assumption~\ref{as:topology}, such as in Case 1 and Case 2, we can construct a virtual graph according to Corollary~\ref{coro} and design the control parameters under this virtual graph. If the system under the virtual graph converges, the system under $\mathcal{G}'$ can also converge.
\begin{figure}[!t]
\centering
\includegraphics[width=2.8in]{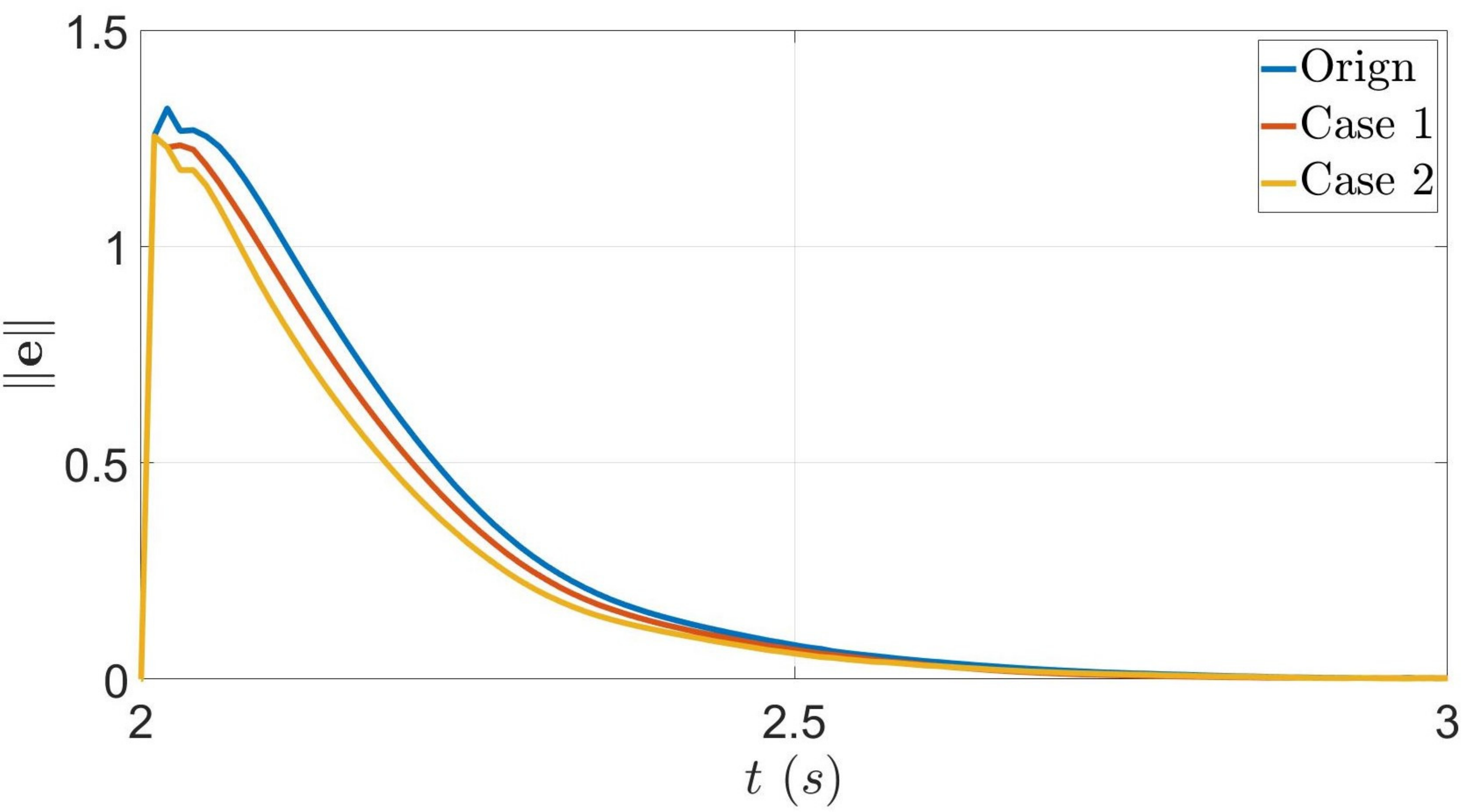}
\caption{Comparison results of $\Vert \mathbf{e} \Vert$ performance in different graph in the second stage.}
\label{difftopo}
\end{figure}

\section{Conclusion}
The attack-resilient distributed secondary control for islanded AC MGs has been explored. This paper presented the conditions for agents with discrete-time dynamics to resiliently track a reference signal propagated by a set of leaders against HDMA with channel noises. The proposed protocol maintains the bounded system stability and guarantees the UUB synchronization performance. In the simulation, we observe that our protocol successfully deceives normal agents into tracking the reference values even against HDMA. Facing with the sparse communication network and wider false data range, the ultimate bound of DACC is within $2 \cdot 10^{-6}$, while MSR-type and hidden layer method fail. Future work includes considering time-delay graphs and more sophisticated attacks in MGs. 
\vspace{-12pt}
\appendix
\section*{Proof of Theorem~1}
\begin{proof}
% Its derivative function is $F_i'(x)=(1-2\sigma x^2)e^{-\sigma x^2}$, with three poles $(0,1)$ and $(\pm \sqrt{\frac{3}{2\sigma}}, -2e^{-\frac{3}{2}})$. So the Lipschitz constant of $F_i$ is $1$.
As long as the system~(\ref{eq:error0}) is asymptotically stable at the equilibrium point $\textbf{e} = \textbf{0}$, the system~(\ref{dynamic}) with \eqref{foll_al} and \eqref{eq:aij} is asymptotically stable at the equilibrium point $x = x_0$. Let $F(k, \textbf{e})=A(k)\textbf{e}$, which is a continuous function of $\textbf{e}$. 

We first validate the boundness of the Jacobian matrix $J(k, \textbf{e})=\frac{ \partial F(k, \textbf{e})}{ \partial \textbf{e} }$ for $\Vert \textbf{e} \Vert \in \mathcal{I}_S$. Let $J_{ij}(k, \textbf{e})$ be the $(i,j)$ entry of $J$. We have
$$\vert J_{ii}(k, \textbf{e}) \vert =\vert\frac{\partial [F(k, \mathbf{e}(k))]_i}{\partial e_i(k)}\vert\leq 1.$$ 
Then we consider the case where $i\neq j$. If $j\notin \mathcal{V}_i^{in}$, we have $J_{ij}(k, \textbf{e})= 0$. If $j \in \mathcal{V}_i^{in}$, it hold that
\begin{equation}\nonumber
\begin{split}
&J_{ij}(k, \textbf{e})= \partial \frac{q_{ij}(k)e_j(k)}{ q_{ij}(k)+\sum_{p \in \mathcal{V}_i^{in}\setminus j}q_{ip}(k)}/ \partial e_j(k) \\
&= \frac{1}{q_i^s(k)^2}[(1+e_j(k))\sum_{p \in \mathcal{V}_i^{in}\setminus j}q_{ip}(k)+q_{ij}(k)]\frac{ \partial q_{ij}(k)}{ \partial e_j(k) },
\end{split}
\end{equation}
where
\begin{equation}
\begin{split}
\frac{ \partial q_{ij}(k)}{ \partial e_j(k)}= & \exp(-(1-\eta)\sigma r_{ij}(k-1)) \cdot\\
&2\eta \sigma \theta_{ij}(k)\exp(-\eta \sigma \theta_{ij}(k)^2). 
\end{split}
\end{equation}
Since $\Vert \textbf{e} \Vert \in \mathcal{I}_S$, we have $r_{ij}(k-1) \leq 4c_n^2$, $q_{ij}(k) \in [\exp(-4\sigma c_n^2),1]$, and $q_i^s(k) \in[ d_i \exp(-4\sigma c_n^2),d_i]$. Due to the boundness of $\theta_{ij}(k)\exp(-\eta \sigma \theta_{ij}(k)^2)$, $\frac{ \partial q_{ij}(k)}{ \partial e_j(k)}$ is bounded. Besides, $J_{ij}$ is differentiable with respect to $e_j$ in $\mathcal{I}_S$. Therefore, $F$ is Lipschitz, and $J$ is bounded and Lipschitz in $\mathcal{I}_S$, uniformly in $k$.

Substituting $ \mathbf{e}=\mathbf{0}$ into $J$, we have $J_{ij}(k,\mathbf{0})< d_i^{-1}$ and thus, all row sum of $J(k,\mathbf{0})$ is less than $1$. According to the sub-multiplicative property, we have $\Vert \mathbf{e}(k+1) \Vert _1=\Vert J(k,\mathbf{0}) \mathbf{e}(k) \Vert_1\leq \Vert J(k,\mathbf{0})\Vert_1 \Vert \mathbf{e}(k) \Vert_1< \Vert \mathbf{e}(k)\Vert_1$, where $\Vert . \Vert _1$ is the induced one norm of a matrix. Then the linear system $ \mathbf{e}(k + 1) = J(k,\mathbf{0}) \mathbf{e}(k)$ is asymptotically stable. Hence, the origin is asymptotically stable at $ \mathbf{e}=\mathbf{0}$ in $\mathcal{I}_S$ in accordance with Lemma~\ref{lemmades}. 
\end{proof}
\vspace{-12pt}
\section*{Proof of Lemma~\ref{xilemma}}
\begin{proof}
% In this proof, the first part derives the relationship between $A(k+h-1)\cdots A(k+1)A(k)$ and $\underline{a_{ij}}(k)$. The second part illustrate that the relationship between $\Vert A(k+h-1)\cdots A(k+1)A(k) \mathbf{e}(k)\Vert$ and $\underline{q_{l}}(k)$.
The $\Xi(k)$ is calculated by induction. Under Assumption~\ref{as:topology}, for any $i \in \mathcal{V}^{(n)}, n \geq 1$, the upper bound of $\Vert A(k+n-1)\cdots A(k+1)A(k)\mathbf{e}(k) \Vert$ is taken from the worst case where there are $f$ misbehaving agents and $d_{max}-f$ normal agents as its neighbors, and only $g$ normal agents are in $\mathcal{V}^{(n-1)}$. In this case, under Assumption~\ref{as:topology} and $\sigma \in \mathcal{I}_\sigma^p$, the lower bound of $a_{ij}(\kappa)$ for all $i,j \in \mathcal{V}_F, \kappa \in [k,k+h-1]$ is taken by $\underline{a_l}(k)$. 

By $S_i^{(n)}$ we denote the $i$-row sum of product $A(k+n-1)\cdots A(k+1)A(k)\mathbf{e}(k)$. If there is only one leader, consider $A(k)\mathbf{e}(k)$. It is immediate that
\begin{equation}\nonumber
S_i^{(1)}= \sum_{j=1}^w a_{ij}(k) \leq (1- \underline{a_{l}}(k))\Vert \mathbf{e}(k)\Vert \qquad \text{if } i \in \mathcal{V}^{(1)}
\end{equation}
and
\begin{equation}\nonumber
S_i^{(1)}\leq \Vert \mathbf{e}(k)\Vert \qquad \text{if } i \notin \mathcal{V}^{(1)}.
\end{equation}
The inequality is due to the fact that if $i \in \mathcal{V}^{(1)}$, the gap between $S_i^{(1)}$ and $1$ is exactly a weight of the edge between a normal follower and a normal leader $a_{il}$.

Consider $A(k)A(k+1)$. If $i\in \mathcal{V}^{(1)}$,
\begin{equation}\nonumber
S_i^{(2)}= \sum_{j=1}^w a_{ij}(k)S_j^{(1)} \leq (1- \underline{a_{l}}(k)) \Vert \mathbf{e}(k)\Vert
\end{equation}

If $i\in \mathcal{V}^{(2)}$, at least $g$ neighbors exist in $\mathcal{V}^{(1)}$. Then
\begin{equation}\nonumber
\Vert \mathbf{e}(k)\Vert-S_i^{(2)}\geq a_{ij}(k+1)(\Vert \mathbf{e}(k)\Vert-S_j^{(1)})\geq g (\underline{a_{l}}(k))^2 \Vert \mathbf{e}(k)\Vert.
\end{equation}

If $i\in \cup_{n=3}^h \mathcal{V}^{(n)}$, we have $a_{ij}(k)=0$ for all $j \in \mathcal{V}^{(1)}$. Then 

\begin{equation}\nonumber
S_i^{(2)}= \sum_{j\in \mathcal{V}^{(1)}} a_{ij}(k)S_j^{(1)}+\sum_{j\in \cup_{n=2}^h \mathcal{V}^{(n)}} a_{ij}(k)S_j^{(1)} \leq \Vert \mathbf{e}(k)\Vert.
\end{equation}

Then we have $S_i^{(2)} \leq (1-g (\underline{a_{l}}(k))^2 \Vert \mathbf{e}(k)\Vert$. The subsequent recursion is the same as Proposition 4.1 in \citep{CSFRL}. We have $$S_i^{(h)}\leq [1- g^{h-1} (\underline{a_{l}}(k))^{h}] \Vert \mathbf{e}(k)\Vert$$.

According to Assumption~\ref{as:topology}, there exist at least $g$ normal leaders. Combined with the Theorem 1 in \citep{FCSGD}, the gap $\Xi$ can be multiplied by $g$. 
\begin{equation}\nonumber
S_i^{(h)}\leq [1- (g\underline{a_{l}}(k))^{h}]\Vert \mathbf{e}(k)\Vert
\end{equation}
Thus the lemma is established.
\end{proof}

% \section*{Proof of Lemma~\ref{mathlemma1}}
% \begin{proof}
% Let $F(x)=(1-x)^{\frac{1}{c}} -( 1-\frac{1}{c}x)$.
% To show \[(1-x)^{\frac{1}{c}} \leq 1-\frac{1}{c}x\]
% it is equivalent to prove $F(x)\leq 0$.
% Obviously $F(0)=0$ and
% \begin{equation*}
% F'(x)=\frac{1}{c}[1-(1-x)^{\frac{1}{c}-1}].
% \end{equation*}
% For any $x\in [0,1)$ and $c \geq 1$, we have $F'(x) \leq 0$.
% Thus $F(x)$ is non-increasing.
% That is for $x\in [0,1)$,
% \[F(x)\leq F(0)=0.\]
% \end{proof}
\vspace{-12pt}
\section*{Proof of Lemma~\ref{mathlemma2}}
\begin{proof}
The function $F_c(x)$ can be reinterpreted as $$F_c(x)=- \frac{c_1}{\sqrt{c_3}}F_i(\sqrt{c_3}x)+\frac{c_2}{\sqrt{c_4}}F_i(\sqrt{c_4}(c_m-x)),$$
where $F_i$ is defined in \eqref{eq:fi}, as depicted in Fig.~\ref{fi}. As shown in Fig.~\ref{fpsi}, the shape of $F_c$ is composed of a $F_i$ flipped based on the $\mathbf{x}$-axis and a $F_i$ flipped based on a vertical line $x=c_m$.
\begin{figure}[hpt]
\centering
\includegraphics[width=2.8in]{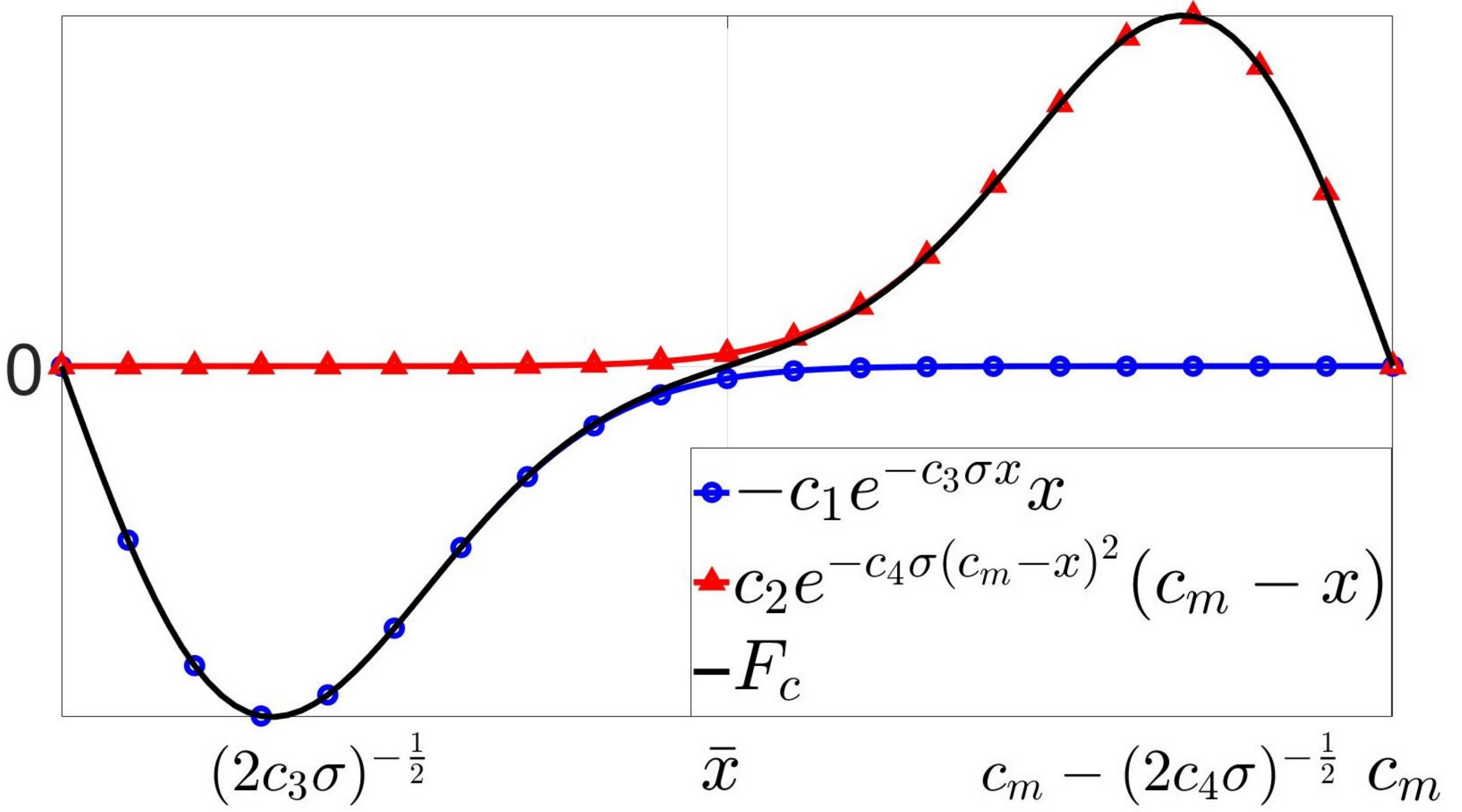}
\caption{The shape of $F_c$ and its components.}
\label{fpsi}
\end{figure}

The derivative function $F_c'(x)$ is 
\begin{equation}\nonumber
\begin{split}
F_c'(x)=&c_1(2 c_3 \sigma x^2-1)\exp(- c_3\sigma x^2)\\
&+c_2 [2c_4 \sigma ( c_m-x)^2-1]\exp(- c_4 \sigma ( c_m-x)^2).
\end{split}
\end{equation}
Denote zero as $z$. Given $x \in [0,c_m]$, there exits two zeros $z_1$ and $z_2$ for $F_c'(x)$. For $x \in [x_1,x_2]$ satisfying $0<x_1<x_2<c_n < z_2$, there are three poles of $F_c(x)$: $x_1$, $z_1$ and $x_2$. Since $z_1$ is the minimum point, $F_c(x) \leq 0$ for all $x \in [x_1,x_2]$ can be ensured by $F_c(x_1) \leq 0$ and $F_c(x_2) \leq 0$. Hence, all we need is to prove $c_n <z_2$. 

We first consider the case where one of the zeros for $F_c'(x)$ is $ (2c_3 \sigma)^{-\frac{1}{2}}$. It is established by $c_m-(2c_4 \sigma)^{-\frac{1}{2}}=(2c_3 \sigma)^{-\frac{1}{2}}$. According to $(2c_3 \sigma)^{-\frac{1}{2}}\geq c_n$, this zero is excluded from $(0,c_n)$ and at most one zero is left in $(0,c_n)$, which means $c_n <z_2$. 

Next, we consider the case where the zeros for $F_c'(x)$ are in $(0,(2c_3 \sigma)^{-\frac{1}{2}})\cup ((2c_3 \sigma)^{-\frac{1}{2}},c_m)$. Then $F_c'(z)=0$ can be equivalent to $F_1(z)=F_2(z)$, where $F_1: [0,(2c_3 \sigma)^{-\frac{1}{2}})\cup ((2c_3 \sigma)^{-\frac{1}{2}},c_m] \rightarrow \mathbb{R}$, and $F_2: [0,c_m] \rightarrow \mathbb{R}_+$. They are expressed as
\begin{equation}\nonumber
\begin{split}
&F_1(x) =-\frac{c_2[2c_4\sigma(c_m-x)^2-1]}{c_1(2c_3\sigma x^2-1)},\\
&F_2(x) =\exp[c_4\sigma(c_m-x)^2-c_3\sigma x^2].
\end{split}
\end{equation}
Upon analysis, it is observed that $F_1'(x)> 0$ for $x$ in either $[0,(2c_3 \sigma)^{-\frac{1}{2}})$ or $((2c_3 \sigma)^{-\frac{1}{2}},c_m]$. Under $c_m-(2c_4 \sigma)^{-\frac{1}{2}}\geq (2c_3 \sigma)^{-\frac{1}{2}}$, $F_1(x)>0$ only if $x$ in either $[0, (2c_3 \sigma)^{-\frac{1}{2}})$ or $(c_m-(2c_4 \sigma)^{-\frac{1}{2}}, c_m]$. Regarding $F_2$, we find that $F_2'(x)<0$ and $F_2(x)>0$ for all $x$ in $[0,c_m]$. Consequently, the solutions for $F_1(x)=F_2(x)$ are only possible in $[0, (2c_3 \sigma)^{-\frac{1}{2}})$ and $(c_m-(2c_4 \sigma)^{-\frac{1}{2}}, c_m]$. Since $c_m-(2c_4 \sigma)^{-\frac{1}{2}}\geq c_n$, we can conclude that there exists at most one solution for $F_1(x)=F_2(x), x \in (0, c_n)$. Therefore, $c_n <z_2$ holds. 
 
The lemma is established.
\end{proof}

\bibliographystyle{elsarticle-num}
\bibliography{mybib_resilient.bib}

\end{document}